\documentclass[draftcls, onecolumn]{IEEEtran}

\usepackage[utf8]{inputenc} 
\usepackage[T1]{fontenc}
\usepackage{url}
\usepackage{ifthen}
\usepackage{cite}
\usepackage[cmex10]{amsmath} 
\usepackage{color}

\usepackage{amssymb,amsthm}
\usepackage{bm, bbm}
\usepackage{mathtools}
\usepackage{commath}
\usepackage{stmaryrd}

\newtheorem{theorem}{Theorem}
\newtheorem{lemma}{Lemma}

\newtheorem{remark}{Remark}

\usepackage{amsmath,amssymb}
\usepackage{xifthen}
\usepackage{mathtools}
\usepackage{enumerate}
\usepackage{microtype}
\usepackage{xspace}
\usepackage{bm}
\usepackage[T1]{fontenc}
\usepackage{fancyhdr}
\usepackage{lastpage}
\usepackage{bbm}

\newcommand{\mbs}[1]{\pmb{#1}}
\newcommand{\vect}[1]{{\lowercase{\mbs{#1}}}}
\newcommand{\mat}[1]{{\uppercase{\mbs{#1}}}}

\newcommand{\T}{{\scriptscriptstyle\mathsf{T}}}

\newtheorem{corollary}{Corollary}
\newtheorem{proposition}{Proposition}

\newcommand{\cond}{\,\vert\,}
\newcommand{\Bigcond}{\,\Big\vert\,}
\renewcommand{\Re}[1][]{\ifthenelse{\isempty{#1}}{\operatorname{Re}}{\operatorname{Re}\left(#1\right)}}
\renewcommand{\Im}[1][]{\ifthenelse{\isempty{#1}}{\operatorname{Im}}{\operatorname{Im}\left(#1\right)}}

\newcommand{\SNR}{\mathsf{snr}}
\newcommand{\snr}{\mathsf{snr}}

\newcommand{\gv}{\vect{g}}
\newcommand{\hv}{\vect{h}}

\newcommand{\qv}{\vect{q}}
\newcommand{\rv}{\vect{r}}

\newcommand{\uv}{\vect{u}}
\newcommand{\vv}{\vect{v}}
\newcommand{\wv}{\vect{w}}
\newcommand{\xv}{\vect{x}}
\newcommand{\yv}{\vect{y}}
\newcommand{\zv}{\vect{z}}

\newcommand{\thetav}{\vect{\theta}}

\newcommand{\Am}{\mat{a}}
\newcommand{\Bm}{\mat{b}}

\newcommand{\Dm}{\mat{d}}

\newcommand{\Fm}{\mat{f}}
\newcommand{\Gm}{\mat{g}}
\newcommand{\Hm}{\mat{h}}
\newcommand{\Jm}{\mat{j}}

\newcommand{\Rm}{\mat{r}}

\newcommand{\Vm}{\mat{V}}

\newcommand{\Xm}{\mat{x}}
\newcommand{\Ym}{\mat{y}}

\newcommand{\Ac}{{\mathcal A}}

\newcommand{\Ic}{{\mathcal I}}

\newcommand{\Kc}{{\mathcal K}}

\newcommand{\Mc}{{\mathcal M}}
\newcommand{\Nc}{{\mathcal N}}

\newcommand{\Xc}{{\mathcal X}}
\newcommand{\Yc}{{\mathcal Y}}

\newcommand{\EE}{\mathbb{E}}

\newcommand{\Id}{\mat{\mathrm{I}}}

\newcommand{\CN}[1][]{\ifthenelse{\isempty{#1}}{\mathcal{N}_{\mathbb{C}}}{\mathcal{N}_{\mathbb{C}}\left(#1\right)}}

\renewcommand{\P}[1][]{\ifthenelse{\isempty{#1}}{\mathbb{P}}{\mathbb{P}\left(#1\right)}}
\newcommand{\E}[1][]{\ifthenelse{\isempty{#1}}{\mathbb{E}}{\mathbb{E}\left[#1\right]}}
\newcommand{\Var}[1][]{\ifthenelse{\isempty{#1}}{\mathsf{Var}}{\mathsf{Var}\left[#1\right]}}
\newcommand{\I}[1][]{\ifthenelse{\isempty{#1}}{\mathbb{I}}{\mathbb{I}\left\{#1\right\}}}
\renewcommand{\det}[1][]{\ifthenelse{\isempty{#1}}{\mathrm{det}}{\mathrm{det}\left(#1\right)}}
\newcommand{\trace}[1][]{\ifthenelse{\isempty{#1}}{\mathrm{tr}}{\mathrm{tr}\left(#1\right)}}
\newcommand{\rank}[1][]{\ifthenelse{\isempty{#1}}{\mathrm{rank}}{\mathrm{rank}\left(#1\right)}}
\newcommand{\diag}[1][]{\ifthenelse{\isempty{#1}}{\mathrm{diag}}{\mathrm{diag}\left(#1\right)}}
\newcommand{\Cov}[1][]{\ifthenelse{\isempty{#1}}{\mathsf{Cov}}{\mathsf{Cov}\left(#1\right)}}

\usepackage{amsmath} 
\usepackage{amssymb} 
\usepackage{textcomp} 
\newcommand{\nr}{n_\text{r}}
\newcommand{\nt}{n_\text{t}}
\renewcommand{\SNR}{\mathsf{snr}}

\newcommand{\Chol}[1]{\mathrm{Chol}\left( #1 \right)}
\newcommand{\Vol}{\mathrm{Vol}}
\newcommand{\Rrho}{\pmb{\Sigma}}
\newcommand{\sign}{\mathrm{sign}}
\newcommand{\ind}{\mathbbm{1}}
\newcommand{\gammasnr}{\gamma}
\renewcommand{\d}{\mathrm{d}}
\newcommand{\corrv}{\qv}
\newcommand{\corr}{q}
\newcommand{\ParaCorr}{\mathcal{Q}}
\newcommand{\Ball}{\mathcal{B}}
\newcommand{\Sphere}{\mathcal{S}}
\newcommand{\rhov}{\pmb{\rho}}

\newcommand{\UT}{\tilde{\mathcal{X}}}

\usepackage{etoolbox}

\newcounter{appendixsection}

\preto\appendices{\setcounter{section}{0}}
\AtBeginEnvironment{appendices}{%
}

\title{Asymptotic Capacity of 1-Bit MIMO\\ Fading Channels}

\author{%
\IEEEauthorblockN{Sheng~Yang and Richard Combes}\\
\IEEEauthorblockA{Laboratoire des Signaux et Systèmes~(L2S)\\
Université Paris-Saclay, CNRS, CentraleSupélec\\
Gif-sur-Yvette, France\\
Email: \{sheng.yang, richard.combes\}@centralesupelec.fr}
\thanks{This work was presented in part at IEEE ISIT 2024 in Athens,
Greece.}
}

\begin{document}

\maketitle

\begin{abstract}
In this work, we investigate the capacity of multi-antenna fading
channels with 1-bit quantized output per receive antenna. Specifically,
leveraging Bayesian statistical tools, we analyze the asymptotic regime
with a large number of receive antennas. In the coherent case, where the
channel state information (CSI) is known at the receiver's side, we
completely characterize the asymptotic capacity and provide the exact
scaling in the extreme regimes of signal-to-noise ratio (SNR) and the
number of transmit antennas. In the non-coherent case, where the CSI is
unknown but remains constant during $T$ symbol periods, we first obtain
the exact asymptotic capacity for $T\le3$. Then, we propose a scheme
involving uniform signaling in the covariance space and derive a
non-asymptotic lower bound on the capacity for an arbitrary block size
$T$. Furthermore, we propose a genie-aided upper bound where the channel
is revealed to the receiver.  We show that the upper and lower bounds
coincide when $T$ is large. In the low SNR regime, we derive the
asymptotic capacity up to a vanishing term, which, remarkably, matches
our capacity lower bound.

\end{abstract}

\section{Introduction}

Recent advances in multiple-antenna communication theory have
significantly impacted modern wireless communication technologies. The
conventional multi-input multi-output (MIMO) channel, typically modeled
as a linear system with additive Gaussian noise, has been extensively
studied. However, in practical scenarios, the resolution of the received
signal is constrained by the quality and precision of the
analog-to-digital converter~(ADC). With low-resolution ADCs, the linear
model becomes invalid, and the channel capacity becomes intractable in
general. This work focuses the capacity of
an extreme regime in which each output is quantized into one bit,
retaining only the sign of the output. This scenario arises due to
constraints on receiver complexity, cost, power consumption, or
bandwidth, especially when outputs are connected to a
baseband-processing unit through rate-limited fronthaul links.  

For single-input single-output~(SISO) channels with quantized output, it has been
shown that the optimal input is discrete. In particular, the capacity of
the 1-bit case is known and the optimal input has no more than three
probability mass points~\cite{singh2009limits, koch2013low,
viterbi2009principles}. The capacity of fading SISO channel with 1-bit
output has been explored in~\cite{mezghani2008analysis,
mezghani2009analysis}. More recently, studies has been extended to the MIMO
case, but only in the low SNR regime~\cite{mezghani2020low}. The problem
becomes more complex in the non-coherent case where the output
distribution relates to the multivariate Gaussian orthant
probability, which is only tractable in low dimensions. 
While non-coherent multiple-input multiple-output~(MIMO) channels with
unquantized output have been extensively studied in the
literature~(e.g., \cite{marzetta1999capacity, Moser,
ZhengTse2002Grassman, durisi2009noncoherent}), especially in the
asymptotic low and high signal-to-noise ratio~(SNR) regimes, the capacity
of channels with quantized output remains largely unexplored. The main
challenges in analyzing the capacity of the 1-bit MIMO channel include
the discrete nature of the output vector and the intractability
of the output distribution given the input. 

In this work, we investigate the capacity of the 1-bit MIMO fading
channels when the number of receive antennas, $\nr$, is large. This
setting is particularly relevant for practical applications, as using
low-resolution ADCs is both cost-effective and power-efficient in
massive MIMO systems.  Our results build on the well known asymptotic
properties of Bayesian statistics from Clarke and Barron~\cite{CB94} in
the large sample size regime, stated as below.  
\setcounter{proposition}{-1}
\begin{proposition}[Clarke and Barron~\cite{CB94}]\label{lemma:CB}
  Let $\yv_1,\ldots,\yv_{\nr}$ be i.i.d.~samples from the
  density $f(\yv;\thetav)$, $\thetav\in\Omega$, satisfying some
  regularity conditions.\footnote{The conditions are specified in
  Appendix~\ref{app:CB_cond}. }
  When $\nr$ is large, for any compact subset $\Kc$ in the interior
  of $\Omega$, we have  
  \begin{equation}
    \max_{\thetav\in\Kc} I(\thetav;
    \pmb{y}_1,\ldots,\pmb{y}_{\nr}) 
    = {\mathrm{dim}(\Kc)\over2} \log {\nr\over 2\pi e} +
    \log \int_{\Kc} \sqrt{\det \pmb{J}_{\theta}(\thetav)} d
    \thetav + o(1), \label{eq:tmp291}
  \end{equation}%
  where $\mathrm{dim}(\Kc)$ is the dimension of the parameter
  space $\Kc$; $\pmb{J}_{\theta}(\thetav)$ is the Fisher
  information matrix with respect to the sample distribution
  parameterized by $\thetav$. The asymptotically optimal distribution of $\thetav$ is  
  \begin{equation}
    p^*(\thetav) = {\sqrt{\det \pmb{J}_{\theta}(\thetav)}
    \over \int_{\Kc} \sqrt{\det \pmb{J}_{\theta}(\thetav)} d
    \thetav}, 
  \end{equation}
  a continuous distribution also known as Jeffreys' prior.  
\end{proposition}
The main contributions of our works are summarized as follows.

First, we consider the coherent case where the channel state
information~(CSI) is perfectly known at the receiver. In this
case, given the channel input, the observations are independent both
in space and time. By carefully examining the Fisher information
of the channel, we completely characterize the asymptotic capacity
up to a vanishing term with $\nr$. Furthermore, we analyze the
scaling of the capacity at high
and low SNR, as well as the case where $\nt$ is large. We shall show
that, essentially, the capacity scales as 
\begin{equation}
  {\nt\over 2} \log {\snr^\beta \nr \over \nt} \label{eq:scaling_coh_app}
\end{equation}%
for some $\beta \in [0,1]$. 

Next, motivated by the difficulty and cost of channel estimation 
with only quantized outputs, we study the non-coherent case where the
channel realization is unknown \emph{a priori}. Specifically, we adopt the
block fading model, where the channel state is unknown but remains
constant during $T$ symbol periods before changing to a new independent
realization\cite{marzetta1999capacity, ZhengTse2002Grassman}. 
We focus on the short coherence time regime, i.e., with $\nt \ge T$. 
Indeed, in this regime, pilot-based communication schemes are
inefficient, necessitating specially designed communication schemes. 
In this case, from each receive antenna, the output distribution of samples
within a block is related to the $T$-dimensional Gaussian orthant probability. 
For $T\le3$, leveraging the closed-form
expression of the distribution, we derive the Fisher information and
obtain the exact asymptotic capacity. For larger values of $T$, as the
orthant probability becomes intractable, we derive upper and lower
bounds on the capacity.  In particular, we propose an input distribution
such that the covariance of the unquantized output is
uniformly distributed in a well chosen parametrization of the space of allowed covariance matrices. The upper bound is based on a
genie-aided argument that reveals the CSI to the receiver while fixing
the input support to be capacity achieving. Further analyses show that
the proposed lower bound is tight in two cases: 1)~when $T$ is large,
and~2)~when the SNR is small. The proposed lower bound suggests the following
scaling
\begin{equation}
  {T-1 \over 4} \log {\gamma^2 \nr \over T},
\end{equation}%
where $\gamma:= {\snr \over 1+\snr}$.

The remainder of the paper is organized as follows. The channel model
and assumptions are introduced in Section~\ref{sec:model}. The coherent
case is studied in Section~\ref{sec:coherent}. The non-coherent case is
examined in Section~\ref{sec:T=23} for $T\le 3$. In
Section~\ref{sec:anyT}, we derive upper and lower bounds valid for any
$T$, while in Section~\ref{sec:asympt} we analyze some asymptotic
regimes. We conclude the paper with discussions in
Section~\ref{sec:discussion}.

\subsubsection*{Notations}
We use $\log$ and $\ln$ for base-$2$ and natural base logarithms,
respectively. 
$\Xm^\T$ denotes the matrix transposition and $\det(\Xm)$ is the
determinant. $\|\cdot\|_{\mathrm{F}}$ is the Frobenius norm.  
$\ind\left( \cdot
\right)$ is the indicator function. 
The binomial coefficient is denoted by $\binom{n}{k} := {n!\over
k!(n-k)!}$. 
We use $\succeq$ to denote the partial
ordering of positive semidefinite matrices.  
$\Chol{\pmb{M}}$ is the Cholesky
decomposition\footnote{Here, for uniqueness, we impose that the diagonal elements of the
decomposed lower triangular matrix are nonnegative.
For positive semidefinite matrices, we apply a limiting argument and
define the Cholesky decomposition as the limit of the decompositions of
a sequence of positive definite matrices. 
} such that $\Chol{\pmb{M}}$ is upper triangular and  
$\Chol{\pmb{M}}^\T \Chol{\pmb{M}} = \pmb{M}$ for any $\pmb{M}\succeq
0$. We use the standard notation $o(\cdot)$ and $O(\cdot)$ for
asymptotically vanishing and bounded terms, respectively. $a\sim b$ means
${a\over b} \to 1$ in the indicated limit. 
$\Gamma(x)$ is the Gamma function, with $\Gamma(x+1)=x \Gamma(x)$ and the well-known
Stirling's approximation $\Gamma(x) \sim \sqrt{2\pi\over x}\left(
{x\over e}
\right)^x$ for large $x$. 
The notation $\pmb{x}\circ\pmb{y}$ stands for the Hadamard product of $\pmb{x}$ and
$\pmb{y}$. 
We use $[T]$ to denote
the set $\{1,\ldots,T\}$. A unit $n$-ball is defined as $\Ball_{n}:=\{\xv\in\mathbb{R}^{n\times1}:
\|\xv\| \le 1\}$. A unit $(n-1)$-sphere is the boundary of a unit $n$-ball,
i.e., $\Sphere_{n-1}:=\{\xv\in\mathbb{R}^{n\times1}:
\|\xv\| = 1\}$. 
We also use $\Ball(\xv,\delta)$ to denote a closed ball centered at
$\xv$ with radius $\delta$. 
The Gaussian density function is $\phi(t):=
{1\over\sqrt{2\pi}} e^{-{t^2\over2}}$, $t\in\mathbb{R}$, and the
Q-function is $Q(x):=\int_x^{\infty} \phi(t) \d t$. 
With a slight abuse of notation, we use the random variable to denote its
distribution if confusion is not likely, e.g., maximizing over a random
variable means over the corresponding distribution. Applying a scalar operator
to any set means the element-wise operation, e.g., $f(\Ac) =
\{f(x):\ x\in\Ac\}$.

\section{Problem formulation} \label{sec:model}

\subsection{Channel model}

We consider a multi-antenna block fading channel where the receiver quantizes each
antenna output into one bit. For simplicity, we assume a real
channel model in which the channel matrix changes independently from one block to another
according to the same distribution while remaining constant during
the same block of $T$ symbol periods. Specifically, the received signal in the $i$th block,
$i=1,\ldots,n$, is the following binary matrix
\begin{equation}
  \pmb{Y}_{\nr\times T}[i] = \sign\left(\pmb{H}_{\nr\times\nt}[i] \pmb{X}_{\nt\times
  T}[i] + \pmb{Z}_{\nr\times T}[i]\right), \label{eq:channel}
\end{equation}%
where $\pmb{H}$ and $\pmb{Z}$ are both matrices with
independent and identically distributed~(i.i.d.)~$\mathcal{N}(0,1)$
entries; and the function $\sign(\cdot)$ is applied on each component. This model is an idealized version of a wireless communication
channel with 1-bit ADC, with perfect
synchronization and analog filtering and passband demodulation.

The input is subject to the peak power constraint such that
\begin{equation}
  \Xc_{\snr} := \{\pmb{X}=[\xv_1\ \cdots\ \xv_T]:\
  \xv_i \in {\sqrt{\snr}} \Ball_{\nt}, i\in[T]\}.
\end{equation}
Due to the normalization, we verify that the average SNR at each receive
antenna is $\snr$. At each coherence block, the receiver observes from each of the $\nr$
antennas a $T$-length 
sequence, denoted by $\pmb{y}_1, \ldots, \pmb{y}_{\nr} \in
\mathcal{Y}^{T\times1}$ with $\mathcal{Y}:=\{-1, 1\}$ such that 
$\pmb{Y}^\T = [\yv_1\ \cdots\ \yv_{\nr}]$.
In this paper, we are interested in the short coherence time regime,
i.e., $\nt\ge T$, in which channel estimation is not straightforward. 
The general case can be treated with extra care as discussed in
Section~\ref{sec:discussion}. 

We assume that the channel state information is unknown at the
transmitter side, with $\Hm^\T = [\hv_1\ \cdots\ \hv_{\nr}]$. In addition, since the channel matrix and noise matrix have
i.i.d.~rows, we have 
\begin{IEEEeqnarray}{rCl}
  p(\Ym, \Hm \cond \Xm) &=& \prod_{i=1}^{\nr} p(\yv_i, \hv_i \cond \Xm). 
\end{IEEEeqnarray}%
Finally, each message is encoded over $n$ blocks, and the
(ergodic)~capacity is defined as the maximum number of bits per channel
use carried by the message that can be decoded by the receiver with an arbitrarily small
probability of error when $n\to\infty$.

\subsection{Coherent communication}

In the coherent case where the receiver knows perfectly
$\Hm$ so that the channel capacity has the following single-letter
form\cite{Cover2006elements}
\begin{equation}
  C = \frac{1}{T} \max_{\Xm\in\Xc_{\snr}} I(\Xm; \Ym, \Hm)\quad\text{bits per channel
  use},
\end{equation}%
since the block fading channel is memoryless stationary. 
Furthermore, given the channel matrix $\Hm$, the channel is also memoryless inside
each coherent block. It follows after simple manipulations that 
\begin{equation}
  C = \max_{\xv\in\sqrt{\snr}\Ball_{\nt}} I(\xv; \yv, \Hm),
\end{equation}
where $\xv\in \sqrt{\snr} \Ball_{\nt}$ and $\yv\in\Yc^{\nr\times1}$ are the input and
output vectors for one channel use, respectively. Therefore, the
capacity in the coherent case does not depend on the coherent block size
$T$, and one can assume $T=1$ without loss of generality. 
Moreover, the channel is characterized by ($\nr$ independent copies of) the
conditional distribution with joint density\footnote{The density is
defined with respect to the product measure $\mu(\Hm) m(\yv)$ where
$\mu$ and $m$ are the Lebesgue measure and the counting measure
respectively.} 
\begin{equation}
 p(y, \hv \cond \xv) = p(\hv) p(y\cond\hv, \xv),
\end{equation}%
where $p(\hv) = \prod_{i=1}^{\nt} \phi(h_i)$, and 
\begin{equation}
  p(y \cond \hv, \xv) =  Q(-y\hv^\T \xv). 
\end{equation}%

\begin{lemma} \label{lemma:121_coh}
  The family of distributions $\{ p(y, \hv \cond
  \xv):\ \xv\in\mathbb{R}^{\nt\times1} \}$ is one-to-one parameterized by
  $\xv\in\mathbb{R}^{\nt\times1}$. 
\end{lemma}

\begin{proof}
 Assume that there exists $\xv'\ne\xv$ such that $p(y,\hv\cond \xv) =
 p(y,\hv\cond \xv')$ for every $(y,\hv)$. Since the Q-function is
 strictly monotone and thus invertible, we must have $\hv^\T\xv =
 \hv^\T\xv'$ for every $\hv$, thus, $\xv = \xv'$. This contradicts the assumption.
\end{proof}

We recall the standard results on the volume of the
unit balls and spheres, as well as the 
scaling at high dimensions with Stirling's approximation~(see, e.g.,
\cite{gradshteyn2014table}). 
\begin{lemma} \label{lemma:vol_ball}
  The volume of a unit $\nt$-ball is 
  \begin{equation}
    \Vol(\Ball_{\nt}) = {\pi^{\nt\over2}\over\Gamma({\nt\over2}+1)}. 
  \end{equation}%
  The volume~(surface area) of a unit $(\nt-1)$-sphere is 
  \begin{equation}
    \Vol(\Sphere_{\nt-1}) = \nt \Vol(\Ball_{\nt}). 
  \end{equation}%
  When $\nt$ is large, the log-volume of the unit $\nt$-ball is 
  \begin{equation}
    \log \Vol(\Ball_{\nt}) = {\nt\over2}\log{2\pi e \over \nt} -
    \log\sqrt{\pi \nt} + o(1). \label{eq:logVol-Ball}
  \end{equation}%
\end{lemma}

\subsection{Non-coherent communication}

In the non-coherent case, $\Hm$ is not known at either the transmitter's
or the receiver's side. The channel capacity becomes
\begin{equation}
  C = \frac{1}{T} \max_{\Xm\in\Xc_{\snr}} I(\Xm; \Ym)\quad\text{bits per channel
  use}.
\end{equation}%
Unlike the coherent case, the channel is not memoryless inside each
coherent block, and we have to consider the matrix input. 
Nevertheless, for a given input $\pmb{X}$, the $\nr$ observation
sequences are independent, and the channel is characterized by ($\nr$ copies of) the conditional distribution 
\begin{IEEEeqnarray}{rCl}
  p({\yv} \cond \Xm) &=& \mathbb{P} \left\{ \sign(\hv^\T \Xm
  + \zv^\T) = \yv^\T  \right\} \\
  &=& \mathbb{P} \left\{ \sign(\mathcal{N}(0, \Id + \Xm^\T \Xm)) =
  \yv  \right\},\label{eq:orthant_pmf}
\end{IEEEeqnarray}%
where $p({\pmb{y}} \cond \pmb{X})$ depends on $\pmb{X}$ only through the
covariance matrix 
\begin{equation}
 \Rm(\Xm)  := {\Id + \pmb{X}^\T \pmb{X}}. 
\end{equation}%
As scaling any unquantized output with a positive value does not
change the sign, we see that $p({\pmb{y}} \cond \pmb{X})$
only depends on the normalized covariance matrix $\Rm_0(\pmb{X})$ 
whose diagonal entries are $1$ and
off-diagonal values are ${\pmb{x}_j^\T  \pmb{x}_k \over
\sqrt{(1+\|\pmb{x}_j\|^2)(1+  \|\pmb{x}_k\|^2)}}$, $\forall\,j\ne
k\in[T]$. 
In the trivial cases with $T=1$ or $\snr = 0$, the
output is independent of the input and is uniformly distributed in
$\mathcal{Y}^{\nr}$. Therefore, the capacity is zero. In the following,
we shall assume that $T\ge 2$ and $\snr>0$.  

Let us introduce some definitions that will be used throughout the
paper. Let $\Ic := \{ \{i_1,i_2\}:\ i_1\ne i_2 \in [T]\}$ be the set of all
possible pairs of distinct elements in $[T]$, with $|\Ic| =
\binom{T}{2}$. Let us define 
the correlation coefficients as 
${\rhov}(\pmb{X}) := [\rho_i(\pmb{X}):\
    i \in\Ic]\in \mathbb{R}^{\binom{T}{2}}$, with 
    \begin{equation} 
      \rho_{i}(\pmb{X}) :=  
      {\pmb{x}_{i_1}^\T  \pmb{x}_{i_2} \over
      \sqrt{(1+\|\pmb{x}_{i_1}\|^2)(1+  \|\pmb{x}_{i_2}\|^2)}},
      \quad \pmb{X} \in \Xc,\
      i=\{i_1, i_2\} \in \Ic, 
      \label{eq:def_rho}
    \end{equation}%
    wherein we can verify by Cauchy-Schwartz and the monotonicity of
    $x\mapsto {x\over1+x}$ that
    \begin{equation}
      \rho_i(\Xm) \le \gammasnr := {\snr\over1+\snr}. 
    \end{equation}%
Further, for $\corrv := [\corr_{i}:\ i\in\Ic] \in
\mathbb{R}^{\binom{T}{2}}$, let us define
the following functions
\begin{IEEEeqnarray}{rCl}
  \tilde{\Rrho}(\corrv) &:=&
  \bigl[ \corr_{\{j,k\}} \ind(j\ne k) \bigr]_{j,k\in[T]}, \\
  {\Rrho}(\corrv) &:=& \Id +
  \tilde{\Rrho}(\corrv). 
\end{IEEEeqnarray}%
For all, $\gamma\in[0,1)$, define 
\begin{IEEEeqnarray}{rCl}
  \ParaCorr^{(T)}_\gamma &:=& \left\{ \corrv \in
  \mathbb{R}^{\binom{T}{2}}:\ \Rrho(\corrv)\succeq (1-\gamma)\Id
  \right\} \\
   &=& \left\{ \corrv \in \mathbb{R}^{\binom{T}{2}}:\
   \Rrho(\gamma^{-1}\corrv)\succeq 0 \right\},
\end{IEEEeqnarray}%
and 
\begin{equation}
  \ParaCorr^{(T)}_1 := \bigcup_{\gamma\in[0,1)} \ParaCorr^{(T)}_\gamma. 
\end{equation}
It follows that $\ParaCorr^{(T)}_\gamma$ is increasing with
$\gamma\in[0,1]$ and that $\Rrho(\ParaCorr^{(T)}_{1})$ is the set of all $T\times T$ positive definite
matrices with diagonal $1$. For brevity, we use $\ParaCorr_{\gamma}$ instead
of $\ParaCorr^{(T)}_\gamma$ unless we need to highlight the dependence on $T$. 
Next, define a subset of the input space with all the upper triangular
matrices with normalized columns and nonnegative diagonals. Specifically,
let
\begin{equation}
  \UT_{\snr} := \left\{ \Xm=[\xv_1\ \cdots\ \xv_T]\in\Xc_{\snr}:\ \|\xv_i\| =
  \sqrt{\snr},
  X_{ij} = X_{ij}\ind(i\le j), X_{jj}\ge0, i\in[\nt], j\in[T] \right\}.
\end{equation}%
The following result establishes different parameterizations of the
family of conditional distributions $\{p({\yv} \cond \Xm):
\Xm\in\Xc_{\snr}\}$. 
\begin{lemma}\label{lemma:param_space}
The parameterization of the family of distributions $\{
\mathrm{sign}(\Nc(0,\Rrho(\corrv))):\ \corrv\in\ParaCorr_1 \}$ is
one-to-one. For any $\gamma\in[0,1)$, the set $\ParaCorr_{\gamma}$ is
compact.  Further, we have 
\begin{IEEEeqnarray}{C}
  \rhov(\UT_{\snr}) = \rhov(\Xc_{\snr}) = \ParaCorr_{\gamma}, \label{eq:121}\\
  \Rm_0(\UT_{\snr}) = \Rm_0(\Xc_{\snr}) = \Rrho(\ParaCorr_\gamma). 
  \label{eq:122}
\end{IEEEeqnarray}%
\end{lemma}
\begin{proof}
  See Appendix~\ref{app:param_space}. 
\end{proof}

\begin{proposition}\label{prop:capa0}
  The capacity of the non-coherent 1-bit MIMO channel with peak power
  constraint is  
  \begin{equation}
    C = {1\over T}\max_{\corrv\in\ParaCorr_\gamma} I(\corrv;
    \yv_1,\ldots,\yv_{\nr}) = {1 \over T}\max_{\Xm\in\UT_{\snr}} I(\Xm;
    \yv_1,\ldots,\yv_{\nr}), \label{eq:tmp1111}
  \end{equation}
  where $\yv_1,\ldots,\yv_{\nr}$ in the above mutual information
  expressions are $\nr$~i.i.d.~samples of the
  random variables 
  $\yv=\mathrm{sign}(\mathcal{N}(0,\Rrho(\corrv)))$ and  
  $\yv=\mathrm{sign}(\mathcal{N}(0,\Rm(\Xm)))$, respectively. 
  If $\corrv^*$ is a maximizer of the first equality in \eqref{eq:tmp1111}, then
  $\Xm^* := {\sqrt{\snr}}\,\Chol{\Rrho(\gamma^{-1}\corrv^*)}$ is an optimal input of the second
  equality\footnote{To be more precise,
  $\Chol{\Rrho(\gamma^{-1}\corrv^*)} \in
  \mathbb{R}^{T\times T}$ corresponds to the upper part of the
  $\nt\times T$ input matrix whose lower part is all $0$ by
  definition of $\UT_{\snr}$.}. 
\end{proposition}
\begin{proof}
  See Appendix~\ref{app:capa0}. 
\end{proof}
Then, the following corollary is straightforward. 
\begin{corollary}
  It is without loss of optimality to consider the signaling with an upper triangular
  matrix $\pmb{X}$ where the columns have a constant magnitude
  $\sqrt{\snr}$. It also implies that using $T$ transmit antennas
  out of $\nt\ge T$ is enough to achieve the capacity. 
\end{corollary}
Note that the fact that using $T$ transmit antennas is enough to
achieve the capacity is the simple consequence of the
output distribution only depending on $\Xm^T \Xm$, which also holds in the unquantized
case~\cite{marzetta1999capacity}. 
Finally, we have the following results on the volume of the parameter space.
\begin{lemma}\label{lemma:regionTheta}
  Let the volume of the region $\ParaCorr^{(T)}_\gamma$ be
  $\Vol(\ParaCorr^{(T)}_\gamma) := \int_{\ParaCorr^{(T)}_\gamma}
  d\corrv$. Then, we have  
  \begin{equation}
    \Vol(\ParaCorr^{(T)}_1) := \int_{\ParaCorr^{(T)}_1} d \corrv = {{\pi}^{\binom{T}{2}+1 \over 2} \Gamma(T) \prod_{j=2}^{T-1}\Gamma({j\over2}) \over 2^{T-1} \left(\Gamma({T+1\over2})\right)^T}, \label{eq:vol}
  \end{equation}%
  and 
  \begin{equation}
    \Vol(\ParaCorr^{(T)}_\gamma) = \gamma^{\binom{T}{2}}
    \Vol(\ParaCorr^{(T)}_\gamma),\quad \gamma\in[0,1].
    \label{eq:vol_gamma}
  \end{equation}%
  In particular, $\Vol(\ParaCorr^{(2)}_1) = 2$,
  $\Vol(\ParaCorr^{(3)}_1) =
  {\pi^2\over2}$, and $\Vol(\ParaCorr^{(4)}_1) =
  {32\over27}\pi^2$. When
  $T$ is large, the log-volume satisfies
  \begin{equation}
    {1\over T} \log \Vol(\ParaCorr^{(T)}_1) = {T-1\over4} \log{ 2\pi \sqrt{e} \over
    T} - {\log e\over 8} + o(1). \label{eq:logVol}
  \end{equation}%
\end{lemma}
\begin{proof}
  See Appendix~\ref{app:regionTheta}, where a more precise approximation
  of $\log \Vol(\ParaCorr^{(T)}_1)$ is derived. 
\end{proof}

\section{Coherent communication} \label{sec:coherent}

In this section, we consider the coherent case. First, we derive the asymptotic capacity when 
$\nr$ is large. Then, we examine the regimes where the SNR and the number of transmit antennas 
can also be large.    
\subsection{Channel capacity}

Let us recall that the coherent channel is characterized by the joint density
$p(y,\hv\cond\xv)$. To derive the capacity, we apply Proposition~\ref{lemma:CB} by
identifying $p(y,\hv\cond \xv)$ with $f(\yv;\thetav)$, $(y,\hv)$
with $\yv$, $\xv$ with $\thetav$, and $\Ball_{\nt}$ with $\Kc$. Let
us develop the Fisher information matrix explicitly as 
\begin{IEEEeqnarray}{rCl}
  \Jm_{y,\hv|\xv}(\xv) &:=& \EE_{y,\hv|\xv} \left[ \nabla_{\xv}(\ln
  p(y,\hv \cond \xv)) \nabla^\T_{\xv}(\ln p(y,\hv \cond \xv))
  \right]\\
  &=& \EE_{y,\hv|\xv} \left[ \nabla_{\xv}(\ln
  p(y \cond \xv,\hv)) \nabla^\T_{\xv}(\ln p(y \cond \xv,\hv))
  \right], 
\end{IEEEeqnarray}%
where the second equality is from $p(y,\hv\cond\xv) = p(\hv)
p(y\cond \xv,\hv)$. The following lemma provides the exact
expression of the determinant of the Fisher information matrix. 
\begin{lemma} \label{lemma:fisher_coh}
  Let us define 
  \begin{equation}
    \xi(s) := {\phi^2(s) \over Q(s)(1-Q(s))},
  \end{equation}%
  and 
  \begin{equation}
    \zeta_k(t) := \E [ S^k \xi(t S)], \quad t\in\mathbb{R},\
    \ k\in\mathbb{Z},\ S\sim\mathcal{N}(0,1).
  \end{equation}%
  Then, we have 
  \begin{equation}
    \det(\Jm_{y,\hv|\xv}(\xv)) = 
    \zeta_0(\|\xv\|)^{\nt-1} \zeta_2(\|\xv\|). \label{eq:det}
  \end{equation}%
\end{lemma}
\begin{proof}
  See Appendix~\ref{app:fisher_coh}. 
\end{proof}

Now, we are ready to present the main result for the coherent case. 
\begin{theorem}\label{thm:coh}
  For any $\snr>0$, the asymptotic capacity of the coherent 1-bit
  MIMO channel when $\nr\to\infty$ is 
  \begin{equation}
    C = {\nt\over2} \log {\nr\over 2\pi e} + \log
    \Vol(\mathcal{B}_{\nt}) + \log \alpha_{\snr,\nt}
    + o(1), \label{eq:capa_coh}
  \end{equation}
  where 
  \begin{equation}
    \alpha_{\snr,\nt} := 
    \int_{0}^{\sqrt{\snr}} \zeta_0(r)^{\nt-1 \over 2} \zeta_2(r)^{1\over2} \nt
    r^{\nt-1} \d r. \label{eq:alpha_coh}
  \end{equation}
  The corresponding optimal input distribution is isotropic~(rotation invariant) with 
  \begin{equation}
    p_{\|\xv\|}(r) \propto \zeta_0(r)^{\nt-1 \over 2} \zeta_2(r)^{1\over2}
    r^{\nt-1}, \quad r\in[0,1].  \label{eq:oi_coh}
  \end{equation}%
\end{theorem}
\begin{proof}
  This theorem is an application of Proposition~\ref{lemma:CB} and
  Lemma~\ref{lemma:fisher_coh}. In the appendix, we check that the
  regularity conditions of Proposition~\ref{lemma:CB} are satisfied. 
  Here, $\Kc = \Xc_{\snr} = \sqrt{\snr}\Ball_{\nt}$ and
  $\mathrm{dim}(\Kc)= \nt$. 
  We can therefore apply \eqref{eq:tmp291} in which the integral
  becomes the following after the change of variable to the spherical
  coordinate
  \begin{IEEEeqnarray}{rCl}
    \int_{\sqrt{\snr}\Ball_{\nt}} \sqrt{
    \det(\Jm_{y,\hv|\xv}(\xv))}
    \d \xv &=& \int_{0}^{\sqrt{\snr}} \int_{{\Sphere_{\nt-1}}}
    \zeta_0(r)^{\nt-1 \over 2} \zeta_2(r)^{1\over2} r^{\nt-1} \d \pmb{\omega} \d r \\
    &=& \Vol(\mathcal{B}_{\nt}) \int_{0}^{\sqrt{\snr}} 
    \zeta_0(r)^{\nt-1 \over 2} \zeta_2(r)^{1\over2} \nt r^{\nt-1} \d r, 
  \end{IEEEeqnarray}%
  where we used the fact that 
  \begin{IEEEeqnarray}{rCl}
    \int_{{\Sphere_{\nt-1}}} \d \pmb{\omega} &=& \Vol({\Sphere_{\nt-1}}) =
    \nt \Vol(\mathcal{B}_{\nt}). 
  \end{IEEEeqnarray}%

  From Proposition~\ref{lemma:CB}, the optimal input distribution is
  $\propto \sqrt{ \det(\Jm_{y,\hv|\xv}(\xv))}$ which only depends
  on $\|\xv\|$ according to \eqref{eq:det}. Therefore, the distribution is rotation invariant. 
  Changing to the spherical coordinate and
  marginalizing, we obtain the distribution on the magnitude of the
  input as \eqref{eq:oi_coh}. 
\end{proof}

\subsection{Asymptotic SNR and $\nt$ regimes}

The main quantity in the capacity expression \eqref{eq:capa_coh} is
$\alpha_{\snr,\nt}$ which can be evaluated numerically. In the
following, we provide the exact equivalent in the extreme SNR and
$\nt$ regimes. First, we present the asymptotic behaviors of $\zeta_0$
and $\zeta_2$. 

\begin{lemma}\label{lemma:zeta}
  When $r\to0$, we have 
\begin{IEEEeqnarray}{rCl}
  \zeta_0(r) = {2\over\pi} + O(r^2), \label{eq:zeta0_r0}\\
\zeta_2(r) = {2\over\pi} + O(r^2).\label{eq:zeta2_r0} 
\end{IEEEeqnarray}%
When $r\to\infty$, we have 
\begin{IEEEeqnarray}{rCl}
  \zeta_0(r) &=& {A_0 \over r}(1 + O(1/r^2)),\label{eq:zeta0_r00}\\ 
\zeta_2(r) &=& {A_2 \over r^3}(1 + O(1/r^2)),\label{eq:zeta2_r00} 
\end{IEEEeqnarray}%
where $A_0:= {1\over\sqrt{2\pi}} \int_{-\infty}^{+\infty} \xi(u) du \approx
0.3842$ and $A_2:= {1\over\sqrt{2\pi}} \int_{-\infty}^{+\infty} \xi(u) u^2 du
\approx 0.1390$.
\end{lemma}
\begin{proof}
  See Appendix~\ref{app:zeta}. 
\end{proof}

Then, applying the above result, we can derive the following two
propositions. 
\begin{proposition}[Low and high SNR regimes] \label{prop:coh_snr}
  The term $\alpha_{\snr,\nt}$ has the following
  asymptotic behaviors. 
  \begin{itemize}
    \item In the low SNR regime such that $\nt\,\snr\to0$, we have 
      \begin{equation}
        {\alpha_{\snr,\nt} \sim 
        \left({2\over\pi}\snr\right)^{\nt\over2}}.
        \label{eq:alpha_coh_lowsnr} 
      \end{equation}
    \item In the high SNR regime such that ${\snr\over\nt}
      \to\infty$, we have 
      \begin{equation}
        \alpha_{\snr,\nt} \sim \begin{cases}  
          \displaystyle \int_{0}^{\infty} \sqrt{\zeta_2(r)} \, \d r, &
          \text{if } \nt=1, \\
          A_0^{\nt-1\over 2}
          A_2^{1 \over 2}  \ln \snr, & \text{
          if } \nt=2, \\ {2 \nt \over \nt-2} A_0^{\nt-1\over 2} A_2^{1
          \over 2}  \snr^{{\nt -2\over 4} 
          }, &\text{ if } \nt \ge 3. \end{cases} \label{eq:coh_highsnr}
        \end{equation}
    \end{itemize}
\end{proposition}
\begin{proof}
  See Appendix~\ref{app:coh_snr}. 
\end{proof}
From the above results, we can obtain some interesting observations.
First, at low SNR, the capacity scales\footnote{It is important to note that in the
considered setting of $\nr\to\infty$, the capacity goes to infinity
even at low SNR. } as ${\nt\over2} \log (\snr\,\nr)$ ignoring other
terms depending on $\nt$. The impact of $\snr$ on the capacity is
through the term $\snr\,\nr$. Then, at high SNR, the behavior of the
asymptotic capacity depends on $\nt$. Remarkably, when $\nt = 1$, the
impact of $\snr$ is bounded since $\alpha_{\snr,\nt}$ converges to a
constant. When $\nt=2$, the impact of $\snr$ becomes $\log \log
\snr$. When $\nt=3$, it becomes ${\nt-2\over4} \log \snr$, such that
the capacity scales approximately as ${\nt\over2} \log (\sqrt{\snr}\,\nr)$.

\begin{proposition}[Large $\nt$ regime] \label{prop:coh_nt}
  When $\nt\to \infty$, the term $\alpha_{\snr,\nt}$ has the following
  asymptotic behaviors. 
  \begin{itemize}
    \item For fixed $\snr$, we have 
      \begin{equation}
        \alpha_{\snr,\nt} \sim 
        {\zeta_0(\sqrt{\snr})^{1\over2} \zeta_2(\sqrt{\snr})^{1\over2} \over
        \zeta_0(\sqrt{\snr}) + {\sqrt{\snr}\over2}
        \zeta_0'(\sqrt{\snr})} \left(\snr\, \zeta_0(\sqrt{\snr})\right)^{\nt
        \over 2}. \label{eq:tmp488}
        \end{equation}
        \item For $\snr$ vanishing with $\nt$ such that $\snr\,\nt\to0$, we have \eqref{eq:alpha_coh_lowsnr}. 
        \item For $\snr$ increasing with $\nt$ such that $\nt/\snr \to0$, we have 
          \begin{equation}
            \alpha_{\snr,\nt} \sim  {2} \sqrt{A_0A_2}   (A_0^2
            \snr)^{{\nt -2\over 4}}. 
          \end{equation}
      \end{itemize}
\end{proposition}
\begin{proof}
  See Appendix~\ref{app:coh_nt}. 
\end{proof}

The following result is straightforward from \eqref{eq:tmp488}, 
\eqref{eq:capa_coh}, and the log-volume of the
unit ball \eqref{eq:logVol-Ball}.
\begin{corollary}
  In the large $\nr$ regime, for a given $\snr>0$ and when
  $\nt\to\infty$, we have 
  \begin{equation}
    C = {\nt\over2} \log \left({\nr\over\nt}\snr\,
    \zeta_0(\sqrt{\snr}) \right) - \log\sqrt{\pi \nt} +
    O(1). 
  \end{equation}%
\end{corollary}

\begin{remark}
  We can verify from \eqref{eq:capa_coh} and the subsequent asymptotic
  analyses in this section that the dominant terms in the capacity can be 
  concisely written as ${\nt\over2}\log {\snr^{\beta}\nr \over \nt}$
  with some $\beta\le1$. Note that
  it is equivalent to the sum capacity of $\nt$ parallel Gaussian
  channels, each with an SNR of ${\snr^{\beta} \over \nt}\nr$. In contrast, the
  capacity of the unquantized MIMO channel is given by~\cite{telatar1999capacity}
  \begin{equation}
    \E \left[ {1\over2}\log\det\left(\Id + {\snr\over\nt}\Hm^\T
    \Hm\right) \right] = {\nt\over 2} \log {\snr\, \nr \over \nt} +
    o(1),
  \end{equation}%
  when $\nr$ is large. 
\end{remark}

\section{Non-coherent communication: Special cases with $T=2, 3$}    \label{sec:T=23}

In the remainder of the paper, we focus on the non-coherent case, which is
more involved than the coherent case due to the complexity of
multi-dimensional orthant probability. We begin by presenting the
following general result, which will be applied to the special cases
in this section as well as to arbitrary $T$ in the next sections. 

\begin{theorem}\label{thm:capa}
  When $\nr$ is large, the capacity of the non-coherent 1-bit MIMO channel with
  $\nt\ge T\ge2$ and peak power $\snr>0$ is given by
  \begin{equation}
    C = {T-1\over4} \log {\nr\over 2\pi e} + {1\over T}\log \alpha_{\SNR,
    T} + o(1)\quad\text{bits per channel use}, \label{eq:capa1} 
  \end{equation}
  where $\alpha_{\SNR, T}$ is a function of $\snr$ and $T$
  \begin{equation}
    \alpha_{\snr,T} := 
    \int_{\ParaCorr_{\gamma}} \sqrt{\det \pmb{J}_{\yv|\corrv}(\corrv)} \d
    \corrv,\label{eq:alpha_nc}
  \end{equation}
  with $\pmb{J}_{\yv|\corrv}(\corrv)$ being the Fisher information
  matrix with respect to the distribution
  $f(\yv;\corrv)=\mathbb{P}\left\{ \mathrm{sign}\left(\mathcal{N}(0,
  \Rrho(\corrv))\right) = \yv \right\}$. The optimal input
  is upper-triangular with normalized columns, i.e., $\Xm =
  \sqrt{\snr}\Chol{\Rrho(\gamma^{-1}\corrv)} \in \UT_{\snr}$ such that  
  \begin{equation}
    p(\corrv) = {1\over\alpha_{\snr,T}}\sqrt{\det
    \pmb{J}_{\yv|\corrv}(\corrv)}, \quad \corrv\in\ParaCorr_\gamma.  
  \end{equation}
\end{theorem}
\begin{proof}
  This is a direct application of Proposition~\ref{lemma:CB}, for which the
  regularity conditions are verified in the appendix. 
\end{proof}

The probability mass function~(pmf) $f(\pmb{y}; \corrv)$ is a
set of multivariate Gaussian orthant probabilities~(see, e.g.,
\cite{plackett1954reduction}) which do not have a closed-form
expression in general, except for $T=2$ and $T=3$. In this section,
we will consider these special cases. 

\subsection{The case $T=2$} \label{sec:T2}
In this case, we have a one-dimensional parameter space
$\ParaCorr_\gamma =
[-\gamma, \gamma]$, and $\pmb{y} = [y_1\ y_2]^\T$. Let $\corr = \corr_{
\{1,2\} } \in \ParaCorr_\gamma$, such that the normalized
covariance matrix has off-diagonal element $\corr$. Thus,
it follows from~\cite{plackett1954reduction} that the probability $\mathbb{P}(y_1\ne y_2) = {1 \over \pi} 
\arccos\left( \corr \right)$. Let us define 
\begin{equation}
  \mu_1(\corr) := \eta(\corr) := {1 \over \pi} \arccos\left( \corr \right), 
  \label{eq:lambda}
\end{equation}
and $\mu_0 := 1-\mu_1$. Then, the pmf can be written
explicitly as  
\begin{equation}
  f(\pmb{y}; \corr) = 
  {\mu_0(\corr) \over2} \ind\left( y_1 = y_2 \right) +
  {\mu_1(\corr)\over2} \ind\left( y_1\ne y_2 \right).
\end{equation}
Note that $\mu_1$ can be considered as a reparameterization of
$f$ since the function in \eqref{eq:lambda} is bijective from
$\ParaCorr_\gamma$ to $\Mc := [\mu_{\min}, 1-\mu_{\min}]$ where
\begin{equation}
  \mu_{\min} := {1 \over \pi} \arccos\left( \gammasnr \right).  \label{eq:lmin}
\end{equation}
Therefore, we have  
\begin{IEEEeqnarray}{rCl}
  \int_{\ParaCorr_\gamma} \sqrt{{J}_{\yv|\corr}({\corr})} d {\corr} &=&
  \int_\Mc \sqrt{{J}_{\yv|\mu_1}({\mu_1})} d {\mu_1} \\
  &=& \int_{\Mc} \mu_1^{-{1\over2}} (1-\mu_1)^{-{1\over2}} d
  {\mu_1}  \label{eq:tmp555}\\
  &=& 4\arccos(\sqrt{\mu_{\min}})-\pi, \label{eq:tmp199}
\end{IEEEeqnarray}%
where the first equality is from the reparameterization and the
change of variables using the Jacobian\footnote{This is also known
as the invariance of Jeffreys' prior to reparameterization~(see, e.g.,
\cite{CB94} and the references therein). }. 
Now, we have the following corollary of Theorem~\ref{thm:capa}
for $T=2$. 
    \begin{corollary}
      When $\nr$ is large and $\nt\ge T=2$, we have 
\begin{equation}
  C = {1\over4} \log {\nr\over 2\pi e} + {1\over 2}\log \alpha_{\SNR,
  2} + o(1)\quad\text{bits per channel use},  
\end{equation}
where       
\begin{equation}
        \alpha_{\snr,2} =  4\arccos\left( \sqrt{{1 \over \pi} \arccos\left( \gammasnr \right) } \right)-\pi,\label{eq:capa_T2} 
      \end{equation}
      which is an increasing function of $\snr$ such that 
      \begin{equation}
        {4\gammasnr \over \pi} \le \alpha_{\snr,2} \le {\pi} \gamma. \label{eq:tmp728} 
      \end{equation}
Moreover, the optimal input signaling is $\Xm =
\sqrt{\snr} \left[ \begin{smallmatrix} 1 & \gamma^{-1}\corr^* \\ 0 &
  \sqrt{1-\gamma^{-2} {\corr^*}^2}\end{smallmatrix} \right]$, such that
  $\mu_1=\eta(\corr^*)$ follows a truncated
  $\mathrm{Beta}({1\over2},{1\over2})$ distribution over $\Mc$.
\end{corollary}
\begin{proof}
  We obtain \eqref{eq:capa_T2} from \eqref{eq:alpha_nc} and
  \eqref{eq:tmp199}. It is increasing since $\arccos$ is decreasing,
  and that $\gamma$ is increasing with $\snr$. 
  The bounds in \eqref{eq:tmp728} are due to the convexity of
  $\arccos\Bigl( \sqrt{{1 \over \pi} \arccos\left( x \right) }
  \Bigr)$ in $[0, 1]$, which implies
  \begin{equation}
    {\pi\over4}(1+{x}) \ge
    \arccos\left( \sqrt{{1 \over \pi} \arccos\left( x \right) }
    \right) \ge {\pi\over4}+{x\over\pi},\quad x\in[0,1],
    \label{eq:delta}
  \end{equation}
  where the upper bound is the line joining the end points, while 
  the right hand side is the supporting line of the function at $x=0$. 
  The optimal input signaling is from the fact that
  $\Xm =
  \sqrt{\snr} \left[ \begin{smallmatrix} 1 & \gamma^{-1}\corr^* \\ 0 &
    \sqrt{1-\gamma^{-2} {\corr^*}^2}\end{smallmatrix} \right]
    =
    \sqrt{\snr} \Chol{\Rrho(\gamma^{-1}\corr^*)}$, and that the optimal distribution of
    $\mu_1:=\eta(\corr^*)$ is a truncated Beta distribution, i.e.,
    $\propto \mu_1^{-{1\over2}} (1-\mu_1)^{-{1\over2}}$
    according to \eqref{eq:tmp555}. 
\end{proof}

\subsection{The case $T=3$}
In the case of $T=3$, we have a three-dimensional parameter space
\begin{equation}
  \ParaCorr_{\gamma} = \{\corrv\in[-\gamma, \gamma]^3:\
  \gamma^{-2}(\corr^2_{12}+\corr^2_{13}+\corr^2_{23})-2\gamma^{-3}\corr_{12}\corr_{23}\corr_{13} \le 1\}, \label{eq:Q_T3}
\end{equation} 
where we use $\corr_{ij}$ to denote $\corr_{ \{i,j\} }$ for brevity. One can
verify that the condition in \eqref{eq:Q_T3} is equivalent to
$\Rrho(\gamma^{-1}\corrv)\succeq 0$ for $T=3$. 
And the corresponding pmf is
\begin{IEEEeqnarray}{rCl}
  f(\pmb{y}; \corrv) &=&
  {\mu_0(\corrv)\over2} \ind\left( y_1=y_2=y_3 \right) +
  {\mu_1(\corrv)\over2} \ind\left( y_1\ne y_2=y_3 \right)
  \\ &&+\>
  {\mu_2(\corrv)\over2} \ind\left( y_2\ne y_1=y_3 \right) +
  {\mu_3(\corrv)\over2} \ind\left( y_3\ne y_1=y_2
  \right), 
\end{IEEEeqnarray}%
where $\mu_i := \mathbb{P}\left\{ y_i \ne y_j = y_k \right\}$
for $\{i,j,k\} = \{1,2,3\}$ and $\mu_0 := 1 - \mu_1 -
\mu_2 - \mu_3$. It follows that $\mu_i+\mu_j =
\mathbb{P}\left\{ y_i\ne y_j \right\} = \eta(\corr_{ij})$ for
any $i\ne j \in \{1, 2, 3\}$. We can then inverse the equations
and show that 
\begin{equation}
  \mu_i(\corrv) =
  {1\over2}(\eta(\corr_{ij})+\eta(\corr_{ik}) -
  \eta(\corr_{jk})),
  \quad \{i,j,k\} = \{1,2,3\}. 
\end{equation}
Note that $[\mu_1,\mu_2,\mu_3]^\T \in \Mc :=
\left\{  \Gm\, \eta(\corrv):\
\corrv\in\ParaCorr_\gamma \right\}$, with $\Gm := {1\over2}[(-1)^{\ind(i =
j)}]_{i,j=1,2,3}$, is a
reparameterization of $f(\pmb{y}; \corrv)$ since the mapping is bijective
from $\ParaCorr_\gamma$ to $\Mc$. Hence, as
for the case $T=2$, we have
\begin{IEEEeqnarray}{rCl}
  \int_{\ParaCorr_\gamma} \sqrt{\det(\pmb{J}_{\yv|\corrv}(\corrv))} d\corrv &=&
  \int_{\Mc} \sqrt{\det(\pmb{J}_{\yv|\pmb{\mu}}(\pmb{\mu}))}
  d\pmb{\mu} \\
  &=& \int_{\Mc} \prod_{i=0}^{3} \mu_i^{-{1\over2}}
  d\pmb{\mu},
\end{IEEEeqnarray}%
where the last equality can be obtained by noticing that 
\begin{equation}
  \pmb{J}_{\yv|\pmb{\mu}}(\pmb{\mu}) =
  [\mu_i^{-1}\ind(i=j)+\mu_0^{-1}]_{i,j=1,2,3},
\end{equation}
and taking the determinant 
\begin{IEEEeqnarray}{rCl}
  \det(\pmb{J}_{\yv|\pmb{\mu}}(\pmb{\mu}) &=&
  \det\left(\diag\left\{ \mu_1^{-1}, \mu_2^{-1}, \mu_3^{-1} \right\}
  + \mu_0^{-1}
  \pmb{1}\pmb{1}^\T\right) \\
  &=& \left(\prod_{i=1}^3 \mu_i^{-1}\right)
  \left( 1+{\mu_0^{-1}}\sum_{i=1}^3 {\mu_i} \right) \\
  &=& \prod_{i=0}^{3} \mu_i^{-{1}}.
\end{IEEEeqnarray}%
Now, we have the asymptotic capacity for $T=3$.  
\begin{corollary}\label{coro:T=3}
  When $\nr$ is large and $\nt\ge T=3$, we have 
  \begin{equation}
    C = {1\over2} \log {\nr\over 2\pi e} + {1\over 3}\log \alpha_{\SNR,
    3} + o(1)\quad\text{bits per channel use},  
  \end{equation}
  where       
  \begin{equation}
    \alpha_{\snr,3} = 
    \int_{\Mc} \prod_{i=0}^{3} \mu_i^{-{1\over2}}
    d\pmb{\mu},\label{eq:capa_T3} 
  \end{equation}
  which is increasing with $\gammasnr$ and bounded as
  \begin{equation}
    {\gammasnr^3 \over 4\pi} \le \alpha_{\snr,3} \le
    {\pi^2}. \label{eq:tmp78} 
  \end{equation}
\end{corollary}
\begin{proof}
  See Appendix~\ref{app:T=3}. 
\end{proof}

\section{Non-coherent communication: Arbitrary $T$} \label{sec:anyT}

   For an arbitrary $T>3$, the orthant probability does not have a
   tractable expression~\cite{plackett1954reduction}. Consequently,
   unlike the cases where $T\le3$, it is not feasible to directly work with the Fisher information matrix to derive the capacity. 
   In this section, we propose lower and upper bounds on the
   capacity. 

   \subsection{Lower bounds}

   To derive a capacity lower bound, we will propose a scheme involving
   a specific input distribution and a simple receiver architecture.
   Although suboptimal, the achievable
   rate of the proposed scheme is easy to analyze and provides a tractable
   lower bound on the capacity. 

   \subsubsection{Input distribution}
   Since the asymptotically optimal input distribution, Jeffreys'
   prior, involves the intractable Fisher information matrix for $T>3$, we adopt a more
   tractable input distribution where $\corrv$ is uniformly
   distributed in $\ParaCorr_\gamma$. The following lemma identifies such an input
   distribution. 
   \begin{lemma}\label{lemma:change_of_variable}
     Let $\Xm \in \UT_{\snr}$ and $\rv_i \in \mathcal{B}_{i-1}$ denote the $i-1$ first entries
     of the $i$-th column of ${\snr}^{-{1\over2}}\Xm$. If   
     \begin{equation}
       p(\Xm) \propto \prod_{i=2}^T (1-\|\rv_i\|^2)^{T-i\over2}
       \ind(\rv_i\in\mathcal{B}_{i-1})\ind(\Xm\in\UT_{\snr}), \label{eq:input_dist} 
     \end{equation}%
     then $\corrv:=\pmb{\rho}(\Xm)$ is uniformly distributed in
     $\ParaCorr_\gamma$. 
   \end{lemma}
   \begin{proof}
     See Appendix~\ref{app:change_of_variable}. 
   \end{proof}

   Let us consider the input distribution \eqref{eq:input_dist} over the set
   $\UT_{\snr}$ of $\nt\times T$ normalized upper-triangular matrices.  
   According to the above lemma, we have $p(\corrv) =
   1/\Vol(\ParaCorr_\gamma)$,
   $\corrv\in\ParaCorr_\gamma$, hence the differential entropy
   \begin{equation}
     h(\corrv) = \log \Vol(\ParaCorr_\gamma) = \binom{T}{2}\log\gamma +
     \log \Vol(\ParaCorr_1). 
   \end{equation}%

   \subsubsection{Receiver}
   Before decoding, we first find an estimate of $\corrv$. To that
   end, note that from \eqref{eq:lambda}, ${\corr}_{i} =   \cos\left( {\pi}
   \mathbb{P}(Y_{ki_1}\ne Y_{ki_2}) \right)$, $\forall\,k\in[\nr]$ and
   $i=\{i_1, i_2\}\in\Ic$. Let
   $p_{i} := \mathbb{P}(Y_{ki_1}\ne Y_{ki_2})$, $i\in\Ic$. 
   Replacing the probability by the sample mean, we have the following
   estimator
   \begin{equation}
     \hat{\corr}_{i} =   \cos\left( {\pi\over \nr}
     \sum_{k=1}^{\nr}  \ind(Y_{ki_1}\ne Y_{ki_2}) \right). 
   \end{equation}%
   Therefore, for a given $\corr_{i}$, $i\in\Ic$, 
   \begin{IEEEeqnarray}{rCl}
     |\hat{\corr}_{i} - \corr_{i}| &=&  \left| \cos\left( {\pi\over \nr}
     \sum_{k=1}^{\nr}  \ind(Y_{ki_1}\ne Y_{ki_2}) \right) - \cos\left( {\pi}
     p_{i} \right) \right| \\
     &\le&  \left| {\pi\over \nr}
     \sum_{k=1}^{\nr}  \ind(Y_{ki_1}\ne Y_{ki_2})  -  {\pi}
     p_{i} \right|,
   \end{IEEEeqnarray}%
   where the inequality is from the fact that $|f(x)-f(y)|\le |x-y| \max_{t} |f'(t)|$ for any differentiable
   function $f$; and $|\cos'(t)| = |\sin(t)| \le 1$. 
   Hence, 
   \begin{IEEEeqnarray}{rCl}
     \E \left\{ |\hat{\corr}_{i} - \corr_{i}|^2 \cond \corr_{i}
     \right\} &\le& {\pi^2} \E
     \left| {1\over \nr} \sum_{k=1}^{\nr}  \ind(Y_{ki_1}\ne Y_{ki_2}) -
     p_{i} \right|^2 \\
     &=& {\pi^2} \mathsf{Var}\left( {1\over \nr}
     \sum_{k=1}^{\nr}  \ind(Y_{ki_1}\ne Y_{ki_2}) \right) \\
     &=& {\pi^2\over\nr} \mathsf{Var}\left(  \ind(Y_{ki_1}\ne
     Y_{ki_2}) \right) \\
     &\le& {\pi^2\over\nr}, 
   \end{IEEEeqnarray}%
   where we used the fact that $(Y_{ki_1}, Y_{ki_2})$ are i.i.d.~over
   $k\in[\nr]$.
   Since the upper bound does not depend on $\corr_{i}$,
   we have the following upper bound on the estimation error 
   \begin{equation}
     \E \left\{ |\hat{\corr}_{i} - \corr_{i}|^2 \right\} \le
     {\pi^2\over\nr}, \quad \forall\, i\in\Ic. \label{eq:tmp122}
   \end{equation}%

   \subsubsection{Achievable rate}

   We are now ready to derive a lower bound on the mutual information.
   For $\Xm\in\UT_{\snr}$ and $\corrv = \rhov(\Xm)$, we have
   \begin{IEEEeqnarray}{rCl}
     I(\Xm; \Ym) &=& I(\corrv; \Ym) \\
     &\ge& I(\corrv; \hat{\corrv})  \label{eq:tmp123}\\
     &=& h(\corrv) - h(\corrv \cond \hat{\corrv})  \\
     &=& h(\corrv) - h(\corrv-\hat{\corrv}\cond \hat{\corrv}) \\
     &\ge& h(\corrv) - h(\corrv-\hat{\corrv}) \label{eq:tmp124}\\
     &\ge& h(\corrv) - \sum_{i\in\Ic} h(\corr_{i} - \hat{\corr}_{i})
     \label{eq:tmp124.5}\\
     &\ge& h(\corrv) - {1\over2} \sum_{i\in\Ic} \log \left(2\pi e
     \E|\corr_{i} - \hat{\corr}_{i}|^2\right) \label{eq:tmp125}\\
     &\ge& h(\corrv) - {1\over2}\binom{T}{2}\log \left( {2\pi^3
     e\over\nr} \right), 
   \end{IEEEeqnarray}%
   where \eqref{eq:tmp123} is from data processing
   $\corrv\leftrightarrow \Ym \leftrightarrow \hat{\corrv}$;
   \eqref{eq:tmp124} is from conditioning reduces entropy;
   \eqref{eq:tmp124.5} is from the sub-additivity of entropy; 
   \eqref{eq:tmp125} is from Gaussian maximizes differential entropy
   under the second moment constraint,
   i.e., $h(X) \le {1\over2} \log \left(2\pi e \E (X^2)\right)$; and the last
   inequality is from \eqref{eq:tmp122}. 
   It leads to the following (non-asymptotic) capacity lower bound. 
   \begin{proposition}[Capacity lower bound] \label{prop:LB}
     When $\nt\ge T$, we have
     \begin{equation}
       {C} \ge {T-1\over4} \log
       { \gamma^2 \over 2\pi^3 e} \nr + {1\over T} \log \Vol(\ParaCorr_1). 
       \label{eq:capa_LB}
     \end{equation}%
   \end{proposition}
   To achieve the above rate, the input matrix has independent
   normalized columns following the distribution~\eqref{eq:input_dist},
   while the receiver preprocesses the signal and treats the channel
   $p(\hat{\corrv}\cond\corrv)$ as additive white Gaussian noise
   channel.

   \subsubsection{Capacity lower bound with independent signaling}

   To generate the input distribution \eqref{eq:input_dist}, one need to
   either sample from spherical distributions of different dimensions
   or uniformly in $\ParaCorr_\gamma$. To reduce the complexity, e.g.,
   for practical applications, we propose a signaling involving only scalar distributions. 

   First, let us define the following hypercube
   \begin{equation}
     \tilde{\ParaCorr}_\gamma := \left\{ \corrv\in\mathbb{R}^{\binom{T}{2}}:\
     |\corr_{i}|
     \le \frac{\gamma}{T-1},\ \forall\,i\in\Ic \right\}. 
   \end{equation}%
   Since each matrix $\Rrho(\gamma^{-1}\corrv)$,
   $\corrv\in\tilde{\ParaCorr}_\gamma$, is diagonally dominant and symmetric, we
   have $\Rrho(\gamma^{-1}\corrv)\succeq0$, which implies that
   $\tilde{\ParaCorr}_\gamma\subseteq \ParaCorr_\gamma$.  

   Then, consider $\corrv$ in which the entries are i.i.d.~over
   $[-{\gamma\over T-1}, {\gamma\over T-1}]$. We have  
   \begin{IEEEeqnarray}{rCl}
     I(\corrv; \pmb{Y}) 
     &=& \sum_{i\in\Ic} I(\corr_{i}; \pmb{Y} \,|\, \{\corr_j\}_{j<i} ) \\ 
     &=& \sum_{i\in\Ic} I(\corr_{i}; \pmb{Y} \,|\, \{\corr_j\}_{j<i} )
     + I(\corr_{i}; \{\corr_j\}_{j<i}) \\
     &=& \sum_{i\in\Ic} I(\corr_{i}; \pmb{Y}, \{\corr_{j}\}_{j<i} ) \\ 
     &\ge& \sum_{i\in\Ic} I(\corr_{i}; \pmb{Y}_{i}),
     \label{eq:para}
   \end{IEEEeqnarray}%
   where $\pmb{Y}_{i}$, $i\in\Ic$, is the submatrix of $\pmb{Y}$ with
   only two columns indicated by the pair $i\in\Ic$; we fix an ordering of
   elements in $\Ic$ (e.g., lexicographical ordering of ordered pairs),
   so that the first equality is from the chain rule with respect to
   that specific ordering; the second equality is from the
   independence between the $\corr_{i}$'s. Therefore, the capacity is
   lower bounded by the sum rate of the $\binom{T}{2}$ ``parallel
   channels'' each one of which is equivalent to a channel with $T=2$~($2$ columns). 

   \begin{proposition}
     The capacity of the non-coherent 1-bit MIMO channel with
     $\nt\ge T$ and $\snr>0$ has the following asymptotic lower bound. 
     \begin{equation}
        C \ge {T-1\over4} \log \left(
{8\gammasnr^2 \over \pi^3 e (T-1)^2 }
 \nr \right) + o(1).
        \label{eq:lb_T}
      \end{equation}
  \end{proposition}
   \begin{proof}
     To maximize the rate of each sub-channel in \eqref{eq:para} between
     $\corr_{i}$ and $\pmb{Y}_{i}$, $i\in\Ic$, we apply the result for $T=2$
     obtained in Sec.~\ref{sec:T2}. The range of 
     $\corr_i$ is reduced from $[-\gamma, \gamma]$ to $[-{\gamma\over
     T-1}, {\gamma\over
     T-1}]$. It is enough to redefine $\mu_{\min} = 
     {1 \over \pi} \arccos\bigl({\gammasnr\over T-1} \bigr)$
     in \eqref{eq:lmin}, and make the same change in the corresponding
     term in \eqref{eq:capa_T2}. 
     Finally, applying \eqref{eq:tmp728} to the modified \eqref{eq:capa_T2}, we obtain the lower bound \eqref{eq:lb_T}. 
   \end{proof}

   \begin{remark}
     Note that the bound \eqref{eq:lb_T} is tighter than
     \eqref{eq:capa_LB} when $T\le3$. When $T>3$, the bound \eqref{eq:capa_LB} is
     strictly better. Intuitively, the independent scalar signaling with
     Jeffreys' prior can still outperform uniform signaling in
     $\ParaCorr_\gamma$ when $T$ is small.  As the independent scalar signaling
     necessitates shrinking the support of each $\corr_i$ by a factor of
     $T-1$, resulting in a loss of a factor $T$ inside the log that
     becomes overwhelming when $T$ grows. The large $T$ regime will be
     discussed in the next section. 
   \end{remark}

   \subsection{Upper bound}

   Since the optimal input is upper-triangular, we assume without loss
   of optimality that $\Xm = [\xv_1\ \cdots\ \xv_T]\in\UT_{\snr}$. To derive an upper
   bound, let us reveal the channel matrix to the receiver. We have
   \begin{IEEEeqnarray}{rCl}
     I(\Xm; \Ym) &\le& I(\Xm; \Ym, \Hm) \\
     &=& I(\Xm; \Ym \cond \Hm) + I(\Xm; \Hm)\\
     &=& I(\Xm; \Ym \cond \Hm) \label{eq:tmp900}\\
     &=& H(\Ym \cond \Hm) - H(\Ym \cond \Xm, \Hm) \\
     &\le& \sum_{i=1}^T H(\yv_i \cond \Hm) - \sum_{i=1}^T H(\yv_i \cond
     \xv_i, \Hm) \label{eq:tmp901}\\
     &=& \sum_{i=1}^T I(\xv_i; \yv_i \cond \Hm) \\
     &=& \sum_{i=1}^T I(\xv_i; \yv_i, \Hm),
   \end{IEEEeqnarray}%
   where \eqref{eq:tmp900} is from the independence between $\Xm$ and
   $\Hm$; \eqref{eq:tmp901} is from $p(\Ym\cond\Xm,\Hm) = \prod_i
   p(\yv_i \cond \xv_i, \Hm)$ with $\Ym = [\yv_1\ \cdots\ \yv_T]$, and the
   sub-additivity of entropy.  
   Note that $I(\xv_i; \yv_i, \Hm)$ is the mutual information of the
   coherent $\nr\times i$ channel since the $\nt-i$ last entries of
   $\xv_i$ is $0$, and the channel has equivalently $i$ transmit
   antennas. In addition, the first $i$ entries form a vector that
   is on the $(i-1)$-sphere $\sqrt{\snr}\mathcal{S}_{i-1}$ in $\mathbb{R}^i$. 
   To further upper bound the mutual information, we need  
   a similar result for the coherent case as in
   Lemma~\ref{lemma:fisher_coh}, but for spherical inputs. 
   \begin{lemma} \label{lemma:fisher_coh_norm}
     Let $\xv\in\sqrt{\snr}\mathcal{S}_{\nt-1}$ be parameterized by the
     $\nt-1$ first entries $\tilde{\xv} :=
     [x_1\ \cdots\ x_{\nt-1}]\in\sqrt{\snr} \mathcal{B}_{\nt-1}$. We have
     \begin{equation}
       \det(\Jm_{y,\hv|{\xv}}(\tilde{\xv})) = {\zeta_0(\sqrt{\snr})^{\nt-1} \over
       1 - \snr^{-1}\|\tilde{\xv}\|^2}.
     \end{equation}%
     Furthermore, 
     \begin{equation}
       \int_{\sqrt{\snr}\mathcal{B}_{\nt-1}}
       \sqrt{\det(\Jm_{y,\hv|{\xv}}(\tilde{\xv}))} \d
       \tilde{\xv} =  {1\over2}\left( \snr\,\zeta_0(\sqrt{\snr}) \right)^{\nt-1\over2}
       \Vol(\mathcal{S}_{\nt-1}). 
       \end{equation}%
     \end{lemma}
     \begin{proof}
       See Appendix~\ref{app:fisher_coh_norm}. 
     \end{proof}
     The following proposition can be obtained with the same steps
     applied in Theorem~\ref{thm:coh}. 
     \begin{proposition} \label{prop:capa_coh_norm}
       The asymptotic capacity of the coherent 1-bit MIMO channel with
       spherical inputs, i.e.,  $\xv\in\sqrt{\snr}\mathcal{S}_{\nt-1}$, is 
       \begin{equation}
         C = {\nt-1 \over2} \log {\snr\,\zeta_0(\sqrt{\snr}) \nr\over 2\pi e} + \log
         {\Vol(\mathcal{S}_{\nt-1})\over2} + o(1).
       \end{equation}
     \end{proposition}
    Applying Proposition~\ref{prop:capa_coh_norm}, we have
     \begin{IEEEeqnarray}{rCl}
       \sum_{i=2}^T \max_{\xv_i} I(\xv_i; \yv_i, \Hm) &=& 
       \sum_{i=2}^T {i-1 \over2} \log {\snr\,\zeta_0(\sqrt{\snr} ) \nr\over 2\pi e} + \log {\Vol(\mathcal{S}_{i-1})\over2} + o(1)\\
       &=& {1\over2}\binom{T}{2} \log {\snr\,\zeta_0(\sqrt{\snr} ) \nr\over 2\pi e} + 
       \sum_{i=2}^T \log {\pi^{i\over2}\over\Gamma({i\over2})} + o(1).
     \end{IEEEeqnarray}%

   Finally, we obtain an upper bound on the asymptotic capacity of the non-coherent channel. 
   \begin{proposition}[Capacity upper bound]
     The capacity of the non-coherent 1-bit MIMO channel is upper
     bounded by 
     \begin{equation}
       C \le {T-1\over4} \log {\snr\,\zeta_0(\sqrt{\snr}) \nr\over 2\pi
       e} + {1\over T} \sum_{i=2}^T \log {\pi^{i\over2}\over\Gamma({i\over2})} + o(1). 
     \label{eq:capa_UB}
     \end{equation}%
   \end{proposition}
   From \eqref{eq:zeta0_r00}, $\zeta_0(\sqrt{\snr})\sim
   A_0{\snr}^{-{1\over2}}$ when $\snr\to\infty$, which implies that the
   above upper bound increases with $\snr$ unboundedly. In contrast, the
   lower bound \eqref{eq:lb_T} depends on $\snr$ only through $\gamma$, which is bounded
   when $\snr\to\infty$.

   \section{Non-coherent communication: Asymptotic $T$ and SNR regimes}
   \label{sec:asympt}

   To get further insights on the scaling of the capacity, we look at
   the two regimes of interest for which we can get the exact scaling:
   1) when $T$ is large, 2) when the SNR is small.

\subsection{Large $T$ regime}

From the capacity lower bound \eqref{eq:capa_LB} and the log volume of
$\ParaCorr_\gamma$ given by \eqref{eq:logVol}, we have 
\begin{equation}
  C \ge {T-1\over4} \log {\gamma^2\nr\over \pi^2 \sqrt{e} T} -
  {\log e\over8} + o(1),
\end{equation}%
when $T$ is large and $\nt\ge T$. 
For the upper bound, we first introduce the following lemma.   
\begin{lemma}\label{lemma:sumGamma}
  When $T$ is large, we have
\begin{equation}
   {1\over T} \sum_{i=2}^T \log {\pi^{i\over2}\over\Gamma({i\over2})} =
   {T-1\over4}\log{2\pi e^{{3\over2}}\over T} + {1\over 8}\log {e\over16} +
   o(1). \label{eq:BG3} 
\end{equation}
\end{lemma}
\begin{proof}
  See Appendix~\ref{app:sumGamma}. 
\end{proof}
Applying \eqref{eq:BG3} to \eqref{eq:capa_UB}, we have  
\begin{equation}
  C \le {T-1\over4} \log {\snr\,\zeta_0(\sqrt{\snr}) \sqrt{e}
  \nr\over T } + {1\over 8}\log {e\over16} + o(1).
\end{equation}%
From the above lower and upper bounds, we have the following scaling.   
   \begin{corollary}[Large $T$ regime]
     When $\nt\ge T$ and $\snr>0$, we have 
     $C\sim {T-1\over4} \log {\nr\over T}$. More precisely, 
     \begin{equation}
       C = {T-1\over4} \left( \log{\nr\over T} + O(1) \right).
     \end{equation}%
   \end{corollary}
   The above result suggests that in the large $T$ regime, the uniform
   distribution of $\corrv \in \ParaCorr_\gamma$ is asymptotically optimal. On the other hand,
   the independent scalar signaling achieves the rate \eqref{eq:lb_T}
   that only scales as ${T-1\over4} \log{\nr\over T^2}$, which is
   strictly suboptimal when $T$ is large.  

\subsection{Low SNR regime}
Next, we consider the low SNR regime. The following reformulation of the
channel distribution is useful and can be verified by simple inspection. 
   \begin{lemma}
     The pmf \eqref{eq:orthant_pmf} can be rewritten as 
     \begin{equation}
       f(\pmb{y}; \corrv) = {2^{-T}\over \sqrt{\det\left(
       \Rrho(\corrv) \right)}} \EE_{\pmb{u}}\left( 
       e^{
       {1\over2}(\pmb{u}\circ\pmb{y})^\T\left(\Id-\left(\Rrho(\corrv)\right)^{-1}\right)(\pmb{u}\circ\pmb{y}))}
       \right), \label{eq:llh}
     \end{equation}%
     where $\pmb{u}$ contains i.i.d.~random variables with the standard
     half normal distribution.\footnote{If $X\sim\mathcal{N}(0,1)$,
     $|X|$ follows the standard half normal distribution.} 
   \end{lemma}
   Note that by definition
      $\Rrho(\corrv)
      = \Id + \tilde{\Rrho}(\corrv)$ with
      $\trace(\tilde{\Rrho}(\corrv)) = 0$. Since $\|\corrv\|=O(\gamma)$,
      applying Taylor's approximation\footnote{We have $\det(\Id + \tilde{\Rrho}(\corrv)) = 1 +
      \trace(\tilde{\Rrho}(\corrv)) + O(\|\corrv\|^2)$ and $(\Id +
      \tilde{\Rrho}(\corrv))^{-1} = \Id - \tilde{\Rrho}(\corrv) +
      O(\|\corrv\|^2)$.} for $\gammasnr$
      small, we have 
   \begin{IEEEeqnarray}{rCl}
     \det(\Rrho(\corrv)) &=& 1 + O(\gammasnr^2),
     \label{eq:tmp112}\\
     \Id-\left(\Rrho(\corrv)\right)^{-1} &=&
      \tilde{\Rrho}(\corrv) + O(\gammasnr^2) = O(\gamma).
      \label{eq:tmp113}
   \end{IEEEeqnarray}%
   It follows that 
   \begin{equation}
     f(\pmb{y}; \corrv) = 2^{-T} (1+O(\gamma)), \label{eq:tmp432} 
   \end{equation}%
   since $\sqrt{\det(\Rrho(\corrv))} =  1+O(\gamma^2)$ and
   \begin{IEEEeqnarray}{rCl}
     \EE_{\pmb{u}}\left( e^{
       {1\over2}(\pmb{u}\circ\pmb{y})^\T\left(\Id-\left(\Rrho(\corrv)\right)^{-1}\right)(\pmb{u}\circ\pmb{y}))}
       \right) &=& \EE_{\pmb{u}}\left( e^{
       {1\over2}\|\pmb{u}\circ\pmb{y}\|^2 O(\gamma)}
       \right) \\
       &=& 1+O(\gamma), 
   \end{IEEEeqnarray}%
   where the last equality is from the dominated convergence theorem.
   Furthermore, one can verify that 
   \begin{IEEEeqnarray}{rCl}
{\partial\over \partial \corr_i}
   \ln \det(\Rrho(\corrv)) &=& 2\left[\Rrho(\corrv)^{-1}\right]_{i_1,i_2} \\
   &=& -2 \corr_i + O(\gamma^2) = O(\gamma), 
   \end{IEEEeqnarray}%
  and 
    \begin{IEEEeqnarray}{rCl}
      {\partial\over \partial \corr_i}\EE_{\pmb{u}}\left( e^{
       {1\over2}(\pmb{u}\circ\pmb{y})^\T\left(\Id-\left(\Rrho(\corrv)\right)^{-1}\right)(\pmb{u}\circ\pmb{y}))}
       \right) &=& \EE_{\pmb{u}}\left( {\partial\over \partial q_i} 
e^{ {1\over2}(\pmb{u}\circ\pmb{y})^\T\left(\Id-\left(\Rrho(\corrv)\right)^{-1}\right)(\pmb{u}\circ\pmb{y}))} \right) \\
&=& \EE_{\pmb{u}}\left( u_{i_1} u_{i_2} y_{i_1} y_{i_2} + O(\gamma)\right)\\ 
&=& {2\over\pi} y_{i_1} y_{i_2} + O(\gamma), 
   \end{IEEEeqnarray}%
   where the first equality is obtained by putting the differentiation
   inside the integral; the second equality is from the differentiation and  Taylor's
   approximation; the last one is from the dominated convergence theorem
   and the fact that for $i\ne j$ we have $\EE(u_i u_j) = (\EE(u_i))^2 = \frac{2}{\pi}$
   for standard half normal random variables.  

   Putting everything together, we have for every $i\in\Ic$ 
   \begin{IEEEeqnarray}{rCl}
     {\partial\over\partial\corr_{i}}\ln f(\pmb{y}; \corrv)
     &=& -{1\over2} {\partial\over\partial\corr_{i}} \ln
     {\det(\Rrho(\corrv))} + {\partial\over\partial\corr_{i}} \ln  
\EE_{\pmb{u}}\left( e^{
{1\over2}(\pmb{u}\circ\pmb{y})^\T\left(\Id-\left(\Rrho(\corrv)\right)^{-1}\right)(\pmb{u}\circ\pmb{y}))}
\right) \\
&=& {2\over\pi} y_{i_1} y_{i_2} + O(\gamma). 
   \end{IEEEeqnarray}%
   Therefore, for $i,j\in\Ic$, the $(i, j)$-th entry of the Fisher
   information matrix $\pmb{J}_{\yv|\corrv}(\corrv)$ is
   \begin{equation}
     {4\over \pi^2}  \EE\left( y_{i_1} y_{i_2} y_{j_1}
     y_{j_2} \right) + O(\gammasnr^2),
   \end{equation}%
   where 
   \begin{equation} \EE\left( y_{i_1} y_{i_2} y_{j_1} y_{j_2} \right) = \ind\left( i=j \right) + \ind\left( i\ne j \right) O(\gammasnr). 
   \end{equation}%
   Indeed, whenever $i\ne j$, the product $y_{i_1} y_{i_2} y_{j_1}
   y_{j_2}$ can be written as a product of $y$'s from two or four coordinates.
   In particular, when $(i_1,i_2,j_1,j_2)$ are all distinct, we have 
   \begin{IEEEeqnarray}{rCl}
     \mathbb{P}\left\{ y_{i_1} y_{i_2} y_{j_1} y_{j_2} = 1\right\} &=&
     \mathbb{P}\left\{ (y_{i_1}, y_{i_2}, y_{j_1}, y_{j_2}) \in \mathcal{O}_1
     \right\} \\
     &=& 8 \left({1\over16}+O(\gamma) \right) \\
     &=& {1\over 2}+O(\gamma), 
   \end{IEEEeqnarray}%
   where $\mathcal{O}_1$ is the set of quadruples in $\{-1,+1\}$ such
   that the product is $1$; the second equality is from
   \eqref{eq:tmp432}, i.e., each orthant has almost the same probability
   when $\gamma\to0$. It follows that $\EE\left( y_i y_j y_{i'} y_{j'}
   \right) = 2 \mathbb{P}\left\{ y_i y_j y_{i'} y_{j'} = 1\right\} - 1 =
   O(\gamma)$. Same arguments apply when only two coordinates are involved. 
   
Therefore, we have 
   $\det(\pmb{J}_{\yv|\corrv}(\corrv)) = \det({4\over\pi^2} 
   (\Id+ \pmb{S}))$ where $\pmb{S}$ has $0$ in the diagonal and
   $O(\gammasnr)$ terms in the off diagonals. Since $\trace(\pmb{S}) = 0$,
   we have $\det(\Id+ \pmb{S}) = 1 + O(\gammasnr^2)$ from Taylor's
   approximation. Thus, 
   \begin{equation}
     \sqrt{\det(\pmb{J}_{\yv|\corrv}(\corrv))} = \left({2\over\pi} 
     \right)^{\binom{T}{2}} (1+O(\gammasnr^2)). 
   \end{equation}%

   And the integral over $\ParaCorr_\gamma$ becomes 
   \begin{IEEEeqnarray}{rCl}
     \int_{\ParaCorr_\gamma} \sqrt{\det(\pmb{J}_{\yv|\corrv}(\corrv))} \d
     \corrv &=&  \int_{\ParaCorr_\gamma} \left({2\over\pi} 
     \right)^{\binom{T}{2}} (1+O(\gammasnr^2)) \d \corrv \\
     &=& \Vol(\ParaCorr_\gamma) \left({2\over\pi} 
     \right)^{\binom{T}{2}} (1+O(\gammasnr^2)).
   \end{IEEEeqnarray}%
   
   Applying Lemma~\ref{lemma:regionTheta}, the following becomes
   straightforward. 
   \begin{proposition}
     For any $\nt\ge T\ge2$, when $\snr$ is small, 
      \begin{equation}
       {C} = {T-1\over4} \log { \gammasnr^2 \over 2\pi^3 e} \nr + {1\over
       T} \log \Vol(\ParaCorr_1) + o(1). 
     \end{equation}%
   \end{proposition}

   Remarkably, the above low SNR capacity coincides with the capacity
   lower bound \eqref{eq:capa_LB} up to a vanishing term. This implies
   that uniform distribution over $\ParaCorr_\gamma$ is asymptotically optimal in
   the low SNR regime as well.

\section{Discussions}\label{sec:discussion}

   In this work, we have investigated the capacity of 1-bit
   MIMO fading channels when the number of receive antennas is large. In this
   regime, the communication problem is intimately connected to the parameter
   estimation problem in the large sample size regime, where the Fisher
   information plays a central role. In particular, the capacity scales
   logarithmically with the number of receive antennas~(samples), and almost 
   linearly with the dimension of the parameter space. Intuitively, the
   dimension of the parameter space represents the degrees of freedom of
   the channel in our framework.

   \subsection{Coherent vs. non-coherent communication}

   Both the coherent and non-coherent cases have been studied. In the
   coherent case, we obtain a complete characterization of the asymptotic capacity
   in \eqref{eq:capa_coh}. In particular, the output distribution is
   one-to-one parameterized by the input $\xv$ such that the degrees of
   freedom is ${\nt\over2}$. Furthermore, for $\nt\ge3$, the SNR gain is
   unbounded and the capacity scales as ${\nt\over2}
   \log{\sqrt{\snr} \nr\over\nt}$ when $\snr$ is large. 

   In the non-coherent case, due to the lack of channel knowledge, the
   output distribution depends on the input only through the normalized
   covariance matrix $\Rm_0(\Xm)$, and the degrees of
   freedom of the channel is reduced to ${T-1\over4} \le {\nt-1\over4}$.
   While the exact asymptotic capacity remains unknown for
   $T>3$, the results for $T\le3$ suggest that the SNR gain is bounded,
   unlike the coherent case. Our results also reveal that uniform
   signaling in the covariance space is near optimal in both the large
   $T$ regime and the low SNR regime.

   \subsection{Open problems and extensions}

   One of the remaining open problems in this work is the exact characterization
   of the asymptotic capacity in the non-coherent case when $T>3$. As mentioned
   earlier, the main obstacle lies in the complexity of the 
   orthant probability in high dimensions. While our genie-aided upper bound,
   which involves providing the CSI to the receiver, matches the exact
   scaling in $T$ observed in our lower bound, it explodes with a
   growing $\snr$, leading to an unbounded gap to the lower bound.
   Closing this gap with tighter bounds remains an
   intriguing future direction. 

   In this work, we focused on the extreme case with one bit per output.
   However, our results can potentially extend to higher resolutions. In
   particular, in the coherent case, the result in
   Theorem~\ref{thm:coh} still holds with a simple redefinition of the function
   $\xi$ according to the multi-bit quantizer. In the non-coherent case,
   however, we need to redefine the parameter space for the multi-bit
   case, as the output distribution depends on the original covariance matrix,
   instead of the normalized one. Indeed, the output is no longer
   invariant to scaling before quantization. As a matter of fact, even
   the unquantized case has not been explored to the best of our knowledge. 

  In the non-coherent case, while the assumption $\nt\ge T$ can be
  relaxed, parameterizing the covariance space becomes more intricate. In
  general, it is not true that signaling with upper
  triangular input matrix is without loss of optimality. The main
  challenge in this case is to identify an adequate one-to-one
  parameterization that allows for computation of the volume of the
  parameter space.  

   For simplicity and tractability, we have studied the real channel
   model. However, the complex channel model with Rayleigh fading can be
   approached with the same methodology. The coherent case is rather
   straightforward due to simple closed-form expression on
   the channel output distribution, akin to the real case. For the
   non-coherent case, one can work with the equivalent real model by
   rearranging the real and imaginary parts of the matrices involved. By
   considering only the real part (or only the imaginary part) of the
   output, we have a real fading model with $\nr$ independent outputs
   given the input matrix $\tilde{\Xm} = \left[\begin{smallmatrix}
     \Re(\Xm) \\ \Im(\Xm) \end{smallmatrix} \right] \in
     \mathbb{R}^{2\nt \times T}$. One can then apply the results for
     the real channel with $\tilde{\nt} = 2\nt$ transmit antennas, which 
     provides a lower bound on the capacity of the complex channel. 

  Finally, throughout this study, we focused on the peak power constraint. 
  It would be interesting to investigate the case with an additional 
  average power constraint, or the case with only average power
  constraint. In the case with a finite number of output bits, it is
  known that average power constraint implies bounded input 
  support~\cite{singh2009limits}. In our case, however, the same
  conclusion may not hold since the number of receive antennas is not
  bounded.

       \appendix
       \subsection{Verification of regularity conditions of Proposition~\ref{lemma:CB} }
For completeness, we include this appendix to verify that the regularity
conditions of Proposition~\ref{lemma:CB} are satisfied. 
It can be skipped without impact on the
understanding of the rest of the paper. 
\subsubsection{Regularity conditions for Proposition~\ref{lemma:CB}}\label{app:CB_cond}
From \cite{CB94}, we know that Proposition~\ref{lemma:CB} holds if the following conditions are
satisfied. 
\begin{enumerate}
  \item[(C1)] The subset $\Kc$ is compact and inside the interior of
    $\Omega$ of dimension $d:=\dim(\Omega)$.  
  \item[(C2)] The parameterization of the family $\{f(\yv; \thetav):\
    \thetav\in\Omega\}$ is one-to-one in $\Omega$.  
  \item[(C3)] The density $f(\yv; \thetav)$ is twice continuously
    differentiable in $\thetav$ for almost every $\yv$. For every
    $\thetav\in\Kc$, there is a
    $\delta = \delta(\thetav)$ so that for each $j,k\in[d]$
    \begin{equation}
      \E \sup_{\thetav'\in \Ball(\thetav,\delta)} \left| {\partial^2
      \over \partial\theta'_j \partial \theta'_k} \ln f(\yv; \thetav')
      \right|^2  \label{eq:tmp329} 
    \end{equation}%
    is finite and continuous as a function of $\thetav$ and for each
    $j\in[d]$  
    \begin{equation}
      \E  \left| {\partial
      \over \partial\theta_j } \ln f(\yv; \thetav)
      \right|^{2+\epsilon}   \label{eq:tmp328}
    \end{equation}%
    is finite and continuous as a function of $\thetav$ for every
    $0\le\epsilon\le\epsilon_0$ for some $\epsilon_0>0$.  
\end{enumerate}

\subsubsection{Verification for the coherent case}\label{app:CB_cond_coh}
In this case, $\Kc = \sqrt{\snr} \Ball_{\nt}$ which is closed and bounded in
$\mathbb{R}^{\nt}$, and therefore compact. Condition C2 is verified
in Lemma~\ref{lemma:121_coh}. In the following,
let us identify $\xv$ with $\thetav$, $(y,\hv)$ with $\yv$, and $p(y,\hv \cond
\xv)$ with $f(\yv; \thetav)$.
To check Condition C3, let us first write the first and second derivatives.  
From $p(y\cond \xv, \hv) = Q(-y \hv^\T \xv)$, we have
\begin{IEEEeqnarray}{rCl}
{\partial \over \partial x_j} \ln p(y,\hv \cond \xv) &=& {y
\phi(\hv^\T \xv) \over  Q(-y \hv^\T \xv)} h_j,
\label{eq:tmp200}\\
{\partial^2 \over \partial x_j \partial x_k} \ln p(y,\hv
\cond \xv)
&=& -{ \phi^2(\hv^\T \xv) \over  Q(-y \hv^\T \xv)^2} h_j h_k - {y  \phi(\hv^\T \xv)
\hv^\T \xv \over  Q(-y \hv^\T \xv)} h_j h_k. \label{eq:tmp201}
\end{IEEEeqnarray}%
Instead of proving that \eqref{eq:tmp329} is finite, we will prove a
stronger version, namely, for a fixed $\delta>0$, 
\begin{equation}
  \E_{\hv} \sup_{\xv\in(\sqrt{\snr}+\delta)\Ball_{\nt}, y\in\Yc} \left| {\partial^2
      \over \partial x_j \partial x_k} \ln p(y,\hv \cond \xv)
      \right|^2  < +\infty. \label{eq:tmp429} 
\end{equation}%
First, note that $Q(-y \hv^\T \xv) \ge
Q( |\hv^\T \xv|)$ and
${\phi(|t|)\over Q(|t|)} \le |t| + c$ for some universal
constant\footnote{For any $t>0$, we have ${\phi(t)\over Q(t)} \le
{1+t^2 \over t}$. We can strengthen it by ${\phi(t)\over
Q(t)}=\ind(t\le1) {\phi(t)\over
Q(t)} +  \ind(t>1) {\phi(t)\over
Q(t)} \le {\phi(1)\over Q(1)} + 1 + t$. } $c>0$. Then,
we deduce that 
  \begin{IEEEeqnarray}{rCl}
  {\phi(\hv^\T \xv) \over Q(\hv^\T \xv)} &\le&
   |\hv^\T \xv| + c \\ 
  &\le& (\sqrt{\snr}+\delta)\|\hv\|  + c, 
  \end{IEEEeqnarray}%
where the last inequality is from Cauchy-Schwartz $|\hv^\T \xv| \le \|\hv\|
\|\xv\|$. Finally, applying $|h_j| \le
\|\hv\|$ and $|A+B|^2 \le 2|A|^2 + 2 |B|^2$, and some simple algebra, we
conclude that $\left| {\partial^2 \over \partial x_j \partial x_k} \ln
p(y,\hv \cond \xv) \right|^2$ is upper bounded by a polynomial function
of $\|\hv\|$. Since the moments of $\|\hv\|$ are finite,
\eqref{eq:tmp429} is proved. With the same approach, \eqref{eq:tmp328}
can be proved. The continuity for both \eqref{eq:tmp329} and
\eqref{eq:tmp328} holds by applying 
some {standard arguments} to the expressions \eqref{eq:tmp200} and \eqref{eq:tmp201}.
Specifically, verifying the continuity of \eqref{eq:tmp201} is
equivalent to taking the limit over $\xv$ towards a specific point
$\xv_0$. Since the term inside the expectation is bounded by an
integrable function as shown above, the continuity of the expectation is
proved by moving the limit inside the expectation by the dominated
convergence theorem. Similarly, to check the continuity
\eqref{eq:tmp200}, it is enough to verify that
the supremum of the set of continuous functions is bounded by an
integrable function and is continuous. The boundedness by an integrable
function has already been shown with \eqref{eq:tmp429}. It remains to
show that $\sup_{\xv'\in\Ball(\xv,\delta)} \bigl| {\partial^2
\over \partial x_j' \partial x_k'} \ln p(y,\hv \cond \xv') \bigr|^2$ is
continuous as a function of $\xv$. 

To that end, we define $f(\xv') := \bigl| {\partial^2 \over \partial
x_j' \partial x_k'} \ln p(y,\hv \cond \xv') \bigr|^2$ which is
continuous in $\mathbb{R}^{\nt}$ and therefore uniformly continuous in
$(\sqrt{\snr}+\delta)\Ball_{\nt}$. Next, let us define $g(\zv) :=
\sup_{\xv'\in\Ball(\zv,\delta)} f(\xv')$. 
Note that $g(\zv) = f(\zv^*)$ for some
$\zv^* \in \Ball(\zv,\delta)$ since $\Ball(\zv,\delta)$ is compact and
$f$ is continuous. Now, fix $\eta>0$ and find $\epsilon>0$ such that for
any $\xv,\xv'$ with $\|\xv-\xv'\|\le\epsilon$, we have
$|f(\xv)-f(\xv')|<\eta$, which is guaranteed by the uniform continuity of
$f$. 
Then, for any $\zv'\in\Ball(\zv,\epsilon)$, 
\begin{IEEEeqnarray}{rCl}
g(\zv') &\ge& f(\zv'-(\zv-\zv^*)) \\
&\ge& f(\zv^*) - \eta \\
&=& g(\zv) - \eta,  
\end{IEEEeqnarray}%
where the first inequality is from the fact that $g(\zv') =
\max_{\zv''\in\Ball(\zv',\delta)} f(\zv'')$ and that
$\|\zv-\zv^*\|\le\delta$ by assumption; the second inequality is from
$\|\zv - \zv'\| \le \epsilon$ by assumption and by the uniform
continuity of $f$. By symmetry, we can use the same argument to prove
$g(\zv) \ge g(\zv')-\eta$, which implies that $|g(\zv) - g(\zv')| \le
\eta$ for any $\|\zv - \zv'\|\le \epsilon$. This proves the continuity of $g(\zv)$.

\subsubsection{Verification for the non-coherent case}\label{app:CB_cond_NC}
In this case, $\Kc = \ParaCorr_\gamma$ which is compact according to
Lemma~\ref{lemma:param_space}. 
Let us identify $\corrv$ with $\thetav$, and $p(\yv \cond
\corrv)$ with $f(\yv; \thetav)$.
Condition C2 is also checked in Lemma~\ref{lemma:param_space}.  

In the following, let us check Condition C3. 
First, rewrite
\begin{equation}
  p(\yv\cond \corrv) = {\int_{\Ac_{\yv}} e^{g_{\corrv}(\wv)} \d \wv
  \over \int_{\mathbb{R}^{T}} e^{g_{\corrv}(\wv)} \d \wv}, 
\end{equation}%
where $\Ac_{\yv} := \{\wv\in\mathbb{R}^T:\ \mathrm{sign}(\wv) = \yv\}$
is the orthant corresponding to $\yv$, and $g_{\corrv}(\wv) := 
-{1\over2} \wv^\T \Rrho(\corrv)^{-1} \wv$. We can check after
some algebra that, for any
$j,k\in\Ic$, 
\begin{IEEEeqnarray}{rCl}
{\partial \over \partial \corr_j} \ln p(\yv \cond \corrv) 
&=& \E \left({\partial\over\partial \corr_j}
g_{\corrv}(\wv) \Bigcond \yv\right)
- \E \left({\partial\over\partial \corr_j}
g_{\corrv}(\wv) \right),
\label{eq:tmp403.5}\\
{\partial^2 \over \partial \corr_j \partial \corr_k} \ln p(\yv
\Bigcond \corrv)
&=&  - \E \left({\partial\over\partial \corr_j}
g_{\corrv}(\wv) \Bigcond \yv \right)\E \left({\partial\over\partial \corr_k}
g_{\corrv}(\wv) \Bigcond \yv \right)  \nonumber \\
&&\> +  \E \left({\partial^2\over\partial \corr_j\corr_k}
g_{\corrv}(\wv) + \left({\partial\over\partial \corr_j}
g_{\corrv}(\wv)\right) 
\left({\partial\over\partial \corr_k}
g_{\corrv}(\wv)\right)
 \Bigcond \yv \right)  \nonumber \\
&&\>  
+ \E \left({\partial\over\partial \corr_j}
g_{\corrv}(\wv)  \right)\E \left({\partial\over\partial \corr_k}
g_{\corrv}(\wv)  \right) \nonumber  \\
&&\> - \E \left({\partial^2\over\partial \corr_j\corr_k}
g_{\corrv}(\wv) + \left({\partial\over\partial \corr_j}
g_{\corrv}(\wv)\right) 
\left({\partial\over\partial \corr_k}
g_{\corrv}(\wv)\right)
  \right), 
\label{eq:tmp405}
\end{IEEEeqnarray}%
where the expectation is with respect to $\wv$, i.e., $\E(\cdot):=
{\int_{\mathbb{R}^T} (\cdot)\,e^{g_{\corrv}(\wv)} \d \wv
  \over \int_{\mathbb{R}^{T}} e^{g_{\corrv}(\wv)} \d \wv}
$ and $\E(\cdot \cond \yv) :=
{\int_{\Ac_{\yv}} (\cdot)\,e^{g_{\corrv}(\wv)} \d \wv
\over \int_{\Ac_{\yv}} e^{g_{\corrv}(\wv)} \d \wv}$. To prove
\eqref{eq:tmp329} and \eqref{eq:tmp328}, it is enough to show that both
\eqref{eq:tmp403.5} and \eqref{eq:tmp405} are bounded for all $\yv$ and 
$\corrv$. To that end, we first derive  
\begin{IEEEeqnarray}{rCl}
  {\partial\over\partial \corr_j} g_{\corrv}(\wv) &=& {1\over2}\,
  \wv^\T \Rrho(\corrv)^{-1} \pmb{\Delta}_{j}
  \Rrho(\corrv)^{-1} \wv, \\
  {\partial^2\over\partial \corr_j\partial \corr_k} g_{\corrv}(\wv)
  &=& - \wv^\T \Rrho(\corrv)^{-1} \pmb{\Delta}_{j}
\Rrho(\corrv)^{-1} \pmb{\Delta}_{k}
  \Rrho(\corrv)^{-1} \wv, 
\end{IEEEeqnarray}%
where $\pmb{\Delta}_j\in\mathbb{R}^{T\times T}$, $j=\{j_1, j_2\} \in \Ic$, has $0$
everywhere except for the entries $(j_1,j_2)$ and $(j_2,j_1)$ that are
both $1$. 

Then, notice that for every $\corrv\in\ParaCorr_\gamma$,
$\Rrho(\corrv)\succeq (1-\gamma)\Id$. Hence, for every
$\corrv'\in\Ball(\corrv,\delta)$, the minimum eigenvalue 
\begin{IEEEeqnarray}{rCl}
\lambda_{\min}( \Rrho(\corrv')) &\ge&
\lambda_{\min}( \Rrho(\corrv)) - \|\tilde{\Rrho}(\corrv-\corrv')\|_{\mathrm{F}} \\
   &=& \lambda_{\min}( \Rrho(\corrv)) - \sqrt{2}\|\corrv-\corrv'\| \\
   &\ge& \lambda_{\min}( \Rrho(\corrv)) - \sqrt{2}\delta,
\end{IEEEeqnarray}%
where the first inequality is from
$\lambda_{\min}(\Am+\Bm)\le\lambda_{\min}(\Am)+\|\Bm\|_{\mathrm{F}}$.
Setting $\delta = {1\over2\sqrt{2}} (1-\gamma)$, we have
$\lambda_{\min}( \Rrho(\corrv'))\ge {1\over2}(1-\gamma)$, and therefore,
$\lambda_{\max}(\Rrho(\corrv')^{-1}) \le {2\over1-\gamma}$.
Next, applying $\uv^\T \Am \vv \le \sigma_{\max}(\Am)\|\uv\| \|\vv\|$
where $\sigma_{\max}$ is the largest singular value, we have  
\begin{IEEEeqnarray}{rCl}
  \left| {\partial\over\partial \corr'_j} g_{\corrv'}(\wv) \right| &\le& {\|\wv\|^2 \over2}\,
  \sigma_{\max}(\Rrho(\corrv')^{-1} \pmb{\Delta}_{j}
  \Rrho(\corrv')^{-1})  \\
  &\le& {2\|\wv\|^2\over(1-\gamma)^2}, \\
  \left| {\partial^2\over\partial \corr'_j\partial \corr'_k}
  g_{\corrv'}(\wv) \right| &\le& 
  \|\wv\|^2 \sigma_{\max}( \Rrho(\corrv')^{-1} \pmb{\Delta}_{j}
\Rrho(\corrv')^{-1} \pmb{\Delta}_{k}
  \Rrho(\corrv')^{-1})\\
  &\le& {8 \|\wv\|^2 \over (1-\gamma)^3}, 
\end{IEEEeqnarray}%
where we have used $\sigma_{\max}(\Am\Bm)\le
\sigma_{\max}(\Am)\sigma_{\max}(\Bm)$, $\sigma_{\max}(\pmb{\Delta}_j)
\le 1$, and $\sigma_{\max}(\Rrho(\corrv')^{-1}) = \lambda_{\max}(\Rrho(\corrv')^{-1}) \le
2(1-\gamma)^{-1}$. Finally, applying the above bounds, we can show that
both \eqref{eq:tmp403.5} and \eqref{eq:tmp405} are bounded by the first
and second moments of $\|\wv\|^2$, which are finite since $\|\wv\|^2$ is
Chi-squared distributed. The proof on the boundedness is complete. The
proof on the continuity follows the same ideas as in the coherent
case and is omitted here.

\subsection{Proof of Lemma~\ref{lemma:param_space}}
\label{app:param_space}

The one-to-one parameterization can be checked by contradiction. 
$p(\yv\cond \corrv)$ be the pmf of
$\mathrm{sign}(\Nc(0,\Rrho(\corrv)))$.  
Assume that there exist some $\corrv' \ne
\corrv \in \ParaCorr_1$, and more specifically with loss of generality, that
$\corr_{12}' \ne
\corr_{12}$, such that $p(\yv\cond \corrv')  = p(\yv\cond \corrv)$.
Marginalizing over $y_3,\ldots,y_T$, we have $p(y_1,y_2\cond
\corr_{12}') = p(y_1,y_2\cond \corr_{12})$, implying $\mathbb{P}(y_1\ne y_2
\cond \corr_{12}') = \mathbb{P}(y_1\ne y_2 \cond \corr_{12})$. Thus, from
Section~\ref{sec:T2}, 
  $\arccos(\corr_{12}') = \arccos(\corr_{12})$. 
This contradicts $\corr_{12}'\ne \corr_{12}$ since $\arccos(\cdot)$ is
invertible. Therefore, $p(\yv\cond \corrv)$ must be uniquely identified
by $\corrv$. 

To show that $\ParaCorr_{\gamma}$ is compact, it is enough to prove that
it is bounded and closed since it is a subset in the Euclidean space.
The boundedness is trivial since we must have $|q_i|\le1$,
$\forall\,i\in\Ic$, to guarantee that $\Rrho(\corrv)\succeq0$. Then, to
show that $\ParaCorr_{\gamma}$ is closed, it is equivalent to showing that
the complement of $\ParaCorr_{\gamma}$, ${\ParaCorr}^c_{\gamma}$, is
open. For each $\corrv\in{\ParaCorr}^c_{\gamma}$, we have
$\lambda_{\min}(\Rrho(\corrv)) \le 1-\gamma-\delta$ for some $\delta>0$.
It can be shown that for every $\corrv'$ such that $\|\corrv'-\corrv\|\le
\delta/(2\sqrt{2})$, $\lambda_{\min}(\Rrho(\corrv')) \le
\lambda_{\min}(\Rrho(\corrv)) +
\|\Rrho(\corrv')-\Rrho(\corrv)\|_{\mathrm{F}} = \lambda_{\min}(\Rrho(\corrv)) +
\sqrt{2}\|\corrv'-\corrv\|\le 1-\gamma-\delta/2$. We just showed that
around every point in ${\ParaCorr}^c_{\gamma}$, we can find a ball with
non-zero radius in ${\ParaCorr}^c_{\gamma}$, implying that
${\ParaCorr}^c_{\gamma}$ is open. 

To prove \eqref{eq:121}, it is enough to show 
\begin{equation}
  \rhov(\UT_{\snr}) \overset{\text{(a)}}{\subseteq} \rhov(\Xc_{\snr})
  \overset{\text{(b)}}{\subseteq} \ParaCorr_\gamma
  \overset{\text{(c)}}{\subseteq} \rhov(\UT_{\snr}). 
\end{equation}%
First, (a) is obvious since $\UT_{\snr}\subseteq \Xc_{\snr}$. To show
(b), let $\Xm\in\Xc_{\snr}$, by definition we have   
\begin{IEEEeqnarray}{rCl}
  \Rrho(\rhov(\Xm)) &=&  \Dm \Xm^\T\Xm\Dm - \diag (\Dm \Xm^\T\Xm\Dm) + \Id \\
  &=& \Dm \Xm^\T\Xm\Dm + \diag \left(1-
  { \|\xv_i\|^2 \over 1+ \|\xv_i\|^2}, i\in[T] \right) \\
  &\succeq& (1-\gamma) \Id, 
\end{IEEEeqnarray}%
where $\Dm := \diag\left( {1\over\sqrt{1+\|\xv_i\|^2}}, i\in[T]
\right)$; the last inequality holds since 
${\|\xv_i\|^2 \over 1+\|\xv_i\|^2} \le \gamma$. 
Therefore, $\rhov(\Xm)\in\ParaCorr_\gamma$ by the definition of
$\ParaCorr_{\gamma}$. To prove (c), fix $\corrv\in\ParaCorr_{\gamma}$, and
let
\begin{equation}
  \Xm = \sqrt{\snr} \Chol{\Rrho(\gamma^{-1}\corrv)}. \label{eq:x4q}
\end{equation}%
As the diagonal of ${\Rrho(\gamma^{-1}\corrv)}$ is
normalized, the columns of $\Xm$ has constant norm $\sqrt{\snr}$, which
implies that $\Xm\in\UT_{\snr}$. In
addition, we have for any $i = \{i_1,i_2\} \in \Ic$, 
$\rho_{i}(\Xm) =  {\pmb{x}_{i_1}^\T \over
\sqrt{1+\|\pmb{x}_{i_1}\|^2}} {\pmb{x}_{i_2} \over
\sqrt{1+ \|\pmb{x}_{i_2}\|^2}} = {\xv_{i_1}^\T \xv_{i_2} \over 1+\snr} =
\corr_{i}$ from \eqref{eq:x4q}.
Therefore, for any $\corrv\in\ParaCorr_\gamma$, we have a
$\Xm\in\UT_{\snr}$
such that $\corrv = \rhov(\Xm) \in \rhov(\UT_{\snr})$. This proves (c) and 
completes the proof of \eqref{eq:121}. 

The two equalities of \eqref{eq:122} are from
\eqref{eq:121} and the fact that $\rhov(\Xm)$ contains the off-diagonals
of ${\Rm_0(\Xm)}$. 
\qed

\subsection{Proof of Proposition~\ref{prop:capa0}}\label{app:capa0}

From the Markov chain $\Xm \leftrightarrow \rhov(\Xm)
\leftrightarrow \Ym$, we have 
\begin{IEEEeqnarray}{rCl}
  I(\pmb{X}; \pmb{Y}) &=& I(\pmb{X}, \rhov(\Xm); \pmb{Y}) \\
  &=& I(\pmb{\rho}(\pmb{X}); \pmb{Y}) \\
  &\le& \max_{\corrv \in \ParaCorr_\gamma} I(\corrv; \pmb{Y}), 
\end{IEEEeqnarray}%
where the last inequality holds since
$\pmb{\rho}(\pmb{X})\in\ParaCorr_\gamma$ and that
$\ParaCorr_\gamma$ is compact. 
Let $\corrv^*$ be the optimizer of the mutual information.
Then, we let $\pmb{X}^* = \sqrt{\snr}\Chol{\Rrho(\gamma^{-1}\corrv^*)}$, and verify
that $\pmb{X}^*\in\UT_{\snr}$ and $\rhov\left( \pmb{X}^* \right) = \corrv^*$ 
since the Cholesky decomposition is
injective, i.e., $\Chol{\Am} \ne \Chol{\Bm}$ if $\Am\ne\Bm$. Hence,
\begin{IEEEeqnarray}{rCl}
  I(\pmb{X}^*; \pmb{Y}) &=& I(\pmb{\rho}(\pmb{X}^*); \pmb{Y}) \\
  &=& I(\corrv^*; \pmb{Y}).
\end{IEEEeqnarray}%
This establishes the lower bound of the capacity, and completes the proof.  
\qed

\subsection{Proof of Lemma~\ref{lemma:regionTheta}}\label{app:regionTheta}

According to Lemma~\ref{lemma:param_space}, for any $\gamma\in[0,1)$, we have
$\ParaCorr_{\gamma} = \rhov(\UT_{\snr})$ with $\gamma =
{\snr\over1+\snr}$. Specifically, let $\Xm\in\UT_{\snr}$ and 
$\rv_j := {\snr}^{-{1\over2}}[X_{1j},\ldots,X_{j-1,j}]\in 
\Ball_{j-1}$ contains the normalized version of the $j-1$ first entries of the $j$-th column of $\Xm$ and 
$X_{jj} = \sqrt{\snr} \sqrt{1-\| \rv_{j} \|^2}$, $j=2,\ldots,T$. 
One can verify that for $i=\{i_1,i_2\}\in\Ic$
\begin{equation}
  q_{i} = {1\over 1+\snr} \sum_{l=1}^{\min\{i_1,i_2\}} X_{l,i_1}
  X_{l,i_2}, 
\end{equation}%
and after some algebra we have
\begin{equation}
  \d \corrv = \gamma^{\binom{T}{2}} 
  \prod_{j=2}^{T} \left({1-\| \rv_{j} \|^2}\right)^{T-j\over2}
  \d {\rv_j}. \label{eq:cov}
\end{equation}%
Therefore, for $\gamma\in[0,1)$, 
\begin{IEEEeqnarray}{rCl}
  \int_{\ParaCorr_{\gamma}} \d \corrv 
  &=& \gamma^{\binom{T}{2}} \prod_{j=2}^{T} \int_{\mathcal{B}_{j-1}} \left({1-\|
  \rv_{j} \|^2}\right)^{T-j\over2} \d {\rv_j} \\
  &{=}& \gamma^{\binom{T}{2}} \prod_{j=2}^{T} \Vol(\mathcal{B}_{j-1}) \EE \left({1-\|
  \rv_{j} \|^2}\right)^{T-j\over2}  \label{eq:Eunif}\\
  &=& \gamma^{\binom{T}{2}} \prod_{j=2}^{T} \Vol(\mathcal{B}_{j-1}) \int_{0}^1 (j-1)
  r^{j-2} ({1-r^2})^{T-j\over2} \d r  \label{eq:dist_r}\\
  &=& \gamma^{\binom{T}{2}} \prod_{j=2}^{T} (j-1) \Vol(\mathcal{B}_{j-1})
  \int_{0}^{\pi/2}  \sin^{j-2}(t) \cos^{T-j+1}(t) \d t  \\
  &=& \gamma^{\binom{T}{2}} \prod_{j=2}^{T} {j-1\over2} \Vol(\mathcal{B}_{j-1})
  \,B\left({j-1\over2}, {T-j+2\over2}\right) \\
  &=& \gamma^{\binom{T}{2}} {{\pi}^{\binom{T}{2}+1 \over 2} \Gamma(T)
  \prod_{j=2}^{T-1}\Gamma({j\over2}) \over 2^{T-1}
  \left(\Gamma({T+1\over2})\right)^T}, \label{eq:tmp322}
\end{IEEEeqnarray}%
where in \eqref{eq:Eunif} the expectation is over a uniform
distribution of $\rv_j$ in $\mathcal{B}_{j-1}$; in
\eqref{eq:dist_r} we use the pdf of the magnitude of a point
inside the ball; the subsequent equalities are from standard
elementary operations where $B(a,b)$ is the Beta
function~\cite{gradshteyn2014table}. Applying the monotone convergence
theorem, we
can show that $\Vol(\ParaCorr_1) =
\lim_{\gamma\to1}\Vol(\ParaCorr_\gamma)$, which proves \eqref{eq:vol}
and \eqref{eq:vol_gamma}. 

When $T$ is large, one can apply Stirling's approximation
\begin{equation}
  \Gamma(x+1) \sim \sqrt{2\pi x} \left( x\over e \right)^x. 
\end{equation}%
Furthermore, note that 
\begin{equation}
  \prod_{j=2}^{T-1} \Gamma(j/2) = {G({T+1\over2}) G({T\over2})
  \over G(1)G({3\over2})}, \label{eq:BG1}
\end{equation}%
where $G(x)$ is the Barnes-G function with $G(x+1) = G(x)
\Gamma(x)$ and $G(1) = 1$ and $G(3/2) =
e^{1\over8}A^{3\over2}2^{1\over24}\pi^{1\over4}$ with $A\approx1.282$
the Glaisher-Kinkelin constant. Using the
equivalent
$G(x+1) \sim A^{-1} {\left(2\pi\right)^\frac{x}{2}
x^{\frac{x^2}{2}-\frac{1}{12}}
e^{-\frac{3x^2}{4}+\frac{1}{12}}}$~\cite{adamchik2014contributions}, and plugging everything back
to \eqref{eq:tmp322}, we obtain
\begin{equation}
  \log\Vol(\ParaCorr_1^{(T)}) =
  {T-1\over4}\log{2\pi\sqrt{e}\over T}  - {\log e
  \over 8} T - {1\over24} \log T  - {7\over2}\log A
  + {1\over3} +{1\over4} \log \pi + {1\over48}\log e + o(1),
\end{equation}%
and less precisely, 
\begin{equation}
  {1\over T}\log\Vol(\ParaCorr_1^{(T)}) =
  {T-1\over4}\log{2\pi\sqrt{e}\over T}  - {\log e
  \over 8} + o(1).
\end{equation}%
\qed

\subsection{Proof of Lemma
\ref{lemma:fisher_coh}}\label{app:fisher_coh}
From $p(y\cond \xv, \hv) = Q(-y \hv^\T \xv)$, we have
\begin{IEEEeqnarray}{rCl}
  \nabla_{\xv}(\ln p(y \cond \xv,\hv))
  &=& {y \phi(\hv^\T \xv) \over  p(y \cond \xv,\hv)}  \hv, 
\end{IEEEeqnarray}%
and
\begin{IEEEeqnarray}{rCl}
  \EE_{y,\hv|\xv} \left[ \nabla_{\xv}(\ln
  p(y \cond \xv,\hv)) \nabla^\T_{\xv}(\ln p(y \cond \xv,\hv))
  \right]
  &=& \EE_{\hv}\left[\sum_{y\in\{1,-1\}} {\phi^2(\hv^\T \xv) \over  p(y \cond
  \xv,\hv)}  \hv \hv^\T \right] \\
  &=& \EE_{\hv} \left[ {\phi^2(\hv^\T \xv) \over Q(\hv^\T
  \xv)(1-Q(\hv^\T \xv))}
  \hv \hv^\T \right]. 
\end{IEEEeqnarray}%
Let $\vv := \xv/\|\xv\|$ and 
$\bar{\Vm}$ be such that the matrix
$\Vm := [\vv\ \bar{\Vm}]$ is
unitary. Then, $\gv := \Vm^\T \hv \sim
\mathcal{N}(0,\Id_{\nt})$ has the same distribution as $\hv$ since
any rotation independent of the white Gaussian vector does not change
the distribution. Then, letting $g_1$ be the first component of
the vector $\gv$, the Fisher information matrix becomes  
\begin{IEEEeqnarray}{rCl}
  \IEEEeqnarraymulticol{3}{l}{\Vm\, \EE_{\gv} \left[ {\phi^2(\|\xv\|g_1) \over
  Q(\|\xv\| g_1)(1-Q(\|\xv\| g_1))}
  \gv \gv^\T \right] \Vm^\T}  \nonumber \\
  &=& \Vm\, \begin{bmatrix} \EE_{\gv} \left[ g_1^2 \xi(\|\xv\| g_1)
    \right] &  \\ & \EE_{\gv} \left[ \xi(\|\xv\| g_1)  \right]
    \Id_{\nt-1} \end{bmatrix} \Vm^\T, \label{eq:tmp101-0} 
\end{IEEEeqnarray}%
where we used the fact that $\gv$ has independent component. Since
$\Vm$ is unitary, we obtain the determinant of the Fisher
information matrix~\eqref{eq:det}. 
\qed

\subsection{Proof of Lemma~\ref{lemma:sumGamma}}
\label{app:sumGamma}

As in \eqref{eq:BG1}, we write 
\begin{equation}
  \prod_{j=2}^{T} \Gamma(j/2) = {G({T+2\over2}) G({T+1\over2})
  \over G(1)G({3\over2})}. \label{eq:BG2}
\end{equation}%
Following almost the same subsequent steps, we can verify
\eqref{eq:BG3}. 

\qed

\subsection{Proof of Lemma~\ref{lemma:zeta}} \label{app:zeta}

When $r \to
0$, applying Taylor's approximation
$\xi(rS) = \xi(0) + O(r^2 S^2)$, and the dominated convergence
theorem~(which allows swapping Taylor's approximation and expectation),
we can show \eqref{eq:zeta0_r0} and \eqref{eq:zeta2_r0} with 
$\xi(0) = 4\phi^2(0) = {2\over\pi}$.

Next, we consider the case $r\to\infty$. 
By a change of variable $u = r s$:
\begin{IEEEeqnarray}{rCl}
  \zeta_0(r) &=&  \int_{-\infty}^{+\infty} {\phi^2(r s) \over
  Q(rs)(1-Q(rs))} \phi(s) \d s \\ 
  &=&  {1 \over r}
  \int_{-\infty}^{+\infty} {\phi^2(u) \over Q(u)(1-Q(u))}
  \phi(u/r) \d u,  		
\end{IEEEeqnarray}%
and, similarly,
\begin{IEEEeqnarray}{rCl}
  \zeta_2(r) &=&  \int_{-\infty}^{+\infty} {\phi^2(r
  s) \over Q(rs)(1-Q(rs))} s^2 \phi(s) \d s \\
  &=&  {1 \over
  r^3} \int_{-\infty}^{+\infty} {\phi^2(u) \over
  Q(u)(1-Q(u))} u^2 \phi(u/r) \d u. 
\end{IEEEeqnarray}%
Therefore, applying $1-{u^2\over2r^2} \le {\phi(u/r)\over\phi(0)} \le
1$, which implies ${\phi(u/r)\over\phi(0)} = 1+O(1/r^2)$, we can show
that $\zeta_0(r) = {A_0 \over r}(1 +
O(1/r^2))$ and $\zeta_2(r) = {A_2 \over r^3}(1 +
O(1/r^2))$.

\qed

\subsection{Proof of Proposition~\ref{prop:coh_snr}}
\label{app:coh_snr}

Let us first consider the regime of $\snr \to 0$. Applying
\eqref{eq:zeta0_r0} and \eqref{eq:zeta2_r0}, we have 
\begin{IEEEeqnarray}{rCl}
  \IEEEeqnarraymulticol{3}{l}{\int_{0}^{\sqrt{\snr}} \zeta_0(r)^{\nt-1
  \over 2} \zeta_2(r)^{1\over2} \nt
  r^{\nt-1} \d r}\\
  &=& \left( {2\over\pi} \right)^{\nt \over 2} \left( \int_{0}^{\sqrt{\snr}} \nt
  r^{\nt-1} \d r \right) (1+O(\sqrt{\snr}^2))^{\nt\over2}
  \\ 
  &=& \left( {2\over\pi} \right)^{\nt \over 2} \snr^{\nt \over 2}
  (1+O(\snr))^{\nt\over2}. 
\end{IEEEeqnarray}
We obtain \eqref{eq:alpha_coh_lowsnr} by noticing that
$(1+O(\snr))^{\nt\over2} \to 0$ when $\nt\snr \to 0$. 

Let us now consider the regime of $\snr \to \infty$. 
From \eqref{eq:zeta0_r00}, as long as ${r^2\over\nt} \to \infty$, we
have  
\begin{equation}
  \zeta_0(r)^{\nt-1\over2} \sim {\left(A_0 \over
  r\right)}^{\nt-1\over2}. \label{eq:tmp445} 
\end{equation}
This allows to compute an equivalent for the
integral\footnote{If $f(t)\sim g(t)$ when $t\to\infty$ and $\int_{a}^x g(t) \d
t$ exists and goes to $\infty$ when $x\to\infty$, then  $\int_{a}^x f(t) \d t \sim \int_{a}^x g(t) \d
t$. }, when $\snr \to \infty$. Note that from \eqref{eq:tmp445} and
\eqref{eq:zeta2_r00}, the integrand of interest is equivalent to
$A_0^{\nt-1\over 2} A_2^{1 \over 2}   \nt r^{{\nt \over 2} - 2}$ when
$r\to\infty$. When $\nt\ge3$, the equivalent is integrable in
$[0,\sqrt{\snr}]$, we have  
\begin{IEEEeqnarray}{rCl}
  \IEEEeqnarraymulticol{3}{l}{\int_{0}^{\sqrt{\snr}} \zeta_0(r)^{\nt-1
  \over 2} \zeta_2(r)^{1\over2} \nt r^{\nt-1} \d r} \\*
  \qquad
  &\sim& A_0^{\nt-1\over 2} A_2^{1 \over 2} \int_{0}^{\sqrt{\snr}}  \nt
  r^{{\nt \over 2} - 2} \d r \\
  &=& 
  {2 \nt \over \nt-2}A_0^{\nt-1\over 2} A_2^{1 \over 2} \snr^{{\nt-2 \over 4}}. 
\end{IEEEeqnarray}%
When $\nt=2$, since the equivalent is not integrable near $0$,
we split the integral into two parts $\int_0^1$ and
$\int_1^{\snr}$ where the first part is bounded and the second
part is equivalent to  
\begin{equation}
  \int_{1}^{\sqrt{\snr}}  \nt r^{{\nt \over 2} - 2} \d r
  \sim  \ln \snr. 
\end{equation}
When $\nt = 1$, the integral converges to a constant
\begin{IEEEeqnarray}{rCl}
  \int_{0}^{\sqrt{\snr}} \zeta_2(r)^{1\over2}  \d r &\sim& 
  \int_{0}^{\infty} \zeta_2(r)^{1\over2}  \d r, 
\end{IEEEeqnarray}%
since we have the integrand decays as ${r^{-{3\over2}}}$
when $r$ is large according to \eqref{eq:zeta2_r00}. 
This completes the proof.
\qed

\subsection{Proof of Proposition~\ref{prop:coh_nt}}
\label{app:coh_nt}
The cases with vanishing and increasing $\snr$ are essentially the same as in
Proposition~\ref{prop:coh_snr}, by letting $\nt\to\infty$. We focus on
fixed $\snr$ in the following. First, rewrite the integral in
\eqref{eq:alpha_coh} as $ \nt \int_{0}^{\sqrt{\snr}} g(r)
e^{(\nt-1)f(r)} \d r$ with $f(r):=\ln\left(r
\sqrt{\zeta_0(r)} \right)$. Since $f(r)$ is strictly
increasing\footnote{It is a simple exercise to check that $f'(r)>0$.} for $r>0$, when
$\nt\to\infty$, applying Laplace's method~\cite{de1981asymptotic}, we have 
\begin{equation}
  \int_{0}^{\sqrt{\snr}} g(r) e^{(\nt-1)f(r)} \d r
  \sim { g(\sqrt{\snr}) \over(\nt-1) f'(\sqrt{\snr})} e^{(\nt-1)
  f(\sqrt{\snr})},    
\end{equation}%
which proves the result in \eqref{eq:tmp488}, after some elementary
computations. 

\qed

\subsection{Proof of Corollary~\ref{coro:T=3}} \label{app:T=3}
First, one can upper bound the integral by the
integration over the whole simplex $\mathcal{P}$. Namely,  
\begin{IEEEeqnarray}{rCl}
  \int_{\Mc} \prod_{i=0}^{3} \mu_i^{-{1\over2}} d\pmb{\mu}
  &\le&  
  \int_{\mathcal{P}} \prod_{i=0}^{3} \mu_i^{-{1\over2}}
  d\pmb{\mu} \\
  &=& {\Gamma({1\over2})^4\over\Gamma(4\times {1\over2})}\\
  &=& \pi^2,
\end{IEEEeqnarray}%
which is known as the normalization factor of the Dirichlet
distribution~\cite{gradshteyn2014table}. Then, to find a lower bound, note that from GM-AM
inequality $\prod_{i=0}^{3}
\mu_i^{-{1\over2}}\ge ({1\over4}\sum_{i=0}^3
\mu_i)^{-2} = 16$. Thus, the
integral is lower bounded by $16\Vol(\Mc)$. We have
\begin{IEEEeqnarray}{rCl}
  \Vol(\Mc) &=& |\det(\pmb{G})| \Vol(\eta(\ParaCorr_\gamma)) \\
  &=& {1\over 2} \int_{\eta(\ParaCorr_\gamma)} \d \eta(\corrv) \\
  &=& {1\over 2} \int_{\ParaCorr_\gamma} \prod_{i} |\eta'(\corr_i)| \d \corrv \\
  &\ge& {1\over2} \left({1\over\pi}\right)^3 \Vol(\ParaCorr_\gamma) \\
  &=& {1\over2} \left({\gamma\over\pi}\right)^3 \Vol(\ParaCorr_1), 
\end{IEEEeqnarray}%
since $\det(\Gm) = {1\over2}$ and $|\eta'(\corr)| =
{1\over\pi} {1\over\sqrt{1-\corr^2}} \ge
{1\over\pi}$. Plugging in $\Vol(\ParaCorr_1^{(3)}) =
{\pi^2\over2}$ from Lemma~\ref{lemma:regionTheta}, we obtain the lower bound. 
\qed

\subsection{Proof of Lemma~\ref{lemma:change_of_variable}} \label{app:change_of_variable}
Note that the probability measure is  
\begin{IEEEeqnarray}{rCl}
  P(\d \corrv) &=& P\left(\prod_{i=2}^T \d \rv_i \right) \\
  &\propto&  \prod_{i=2}^T (1-\|\rv_i\|^2)^{T-i\over2} \d \rv_i \\
  &=& \d \corrv,
\end{IEEEeqnarray}%
where the last equality is from the change of
variable~\eqref{eq:cov}. We see that the probability density
is constant, i.e., $\corrv$ follows a uniform distribution. 
\qed

\subsection{Proof of Lemma~\ref{lemma:fisher_coh_norm}}
\label{app:fisher_coh_norm}
Since $\|\xv\|^2$ is constant, we have  
\begin{equation}
  \d \xv = \underbrace{\begin{bmatrix} \Id_{\nt-1} \\ -{1\over x_{\nt}}
    \tilde{\xv}\T \end{bmatrix}}_{\Fm} \d \tilde{\xv} =
    \bar{\Vm} \d \tilde{\xv}. \label{eq:tmp333} 
\end{equation}%
Therefore, we have 
\begin{IEEEeqnarray}{rCl}
  \Jm_{y,\hv|{\xv}}(\tilde{\xv}) 
  &=& \Fm^\T \Jm_{y,\hv|{\xv}}({\xv}) \Fm \label{eq:tmp100}\\
  &=& \Fm^\T \Vm\, \begin{bmatrix} \EE_{\gv} \left[ g_1^2
    \xi(\sqrt{\snr} g_1) \right] &  \\ & \EE_{\gv} \left[ \xi(\sqrt{\snr}
    g_1)  \right] \Id_{\nt-1} \end{bmatrix} \Vm^\T \Fm
    \label{eq:tmp101} \\
    &=& \left[0\ \Fm^\T \bar{\Vm} \right]\, \begin{bmatrix} \EE_{\gv} \left[
      g_1^2 \xi(\sqrt{\snr} g_1) \right] &  \\ & \EE_{\gv} \left[
      \xi(\sqrt{\snr} g_1)  \right] \Id_{\nt-1} \end{bmatrix} \left[0\
      \Fm^\T \bar{\Vm} \right]^\T \label{eq:tmp102} \\
      &=& \EE_{\gv} \left[ \xi(\sqrt{\snr} g_1) \right] \Fm^\T
      \bar{\Vm} \bar{\Vm}^\T \Fm \\
      &=& \EE \left[ \xi(\sqrt{\snr} S) \right] \Fm^\T \Fm
      \label{eq:tmp103}\\
      &=& \EE \left[ \xi(\sqrt{\snr} S) \right] \left( \Id +
      {1\over x_{\nt}^2} \tilde{\xv} \tilde{\xv}^\T \right),
\end{IEEEeqnarray}%
where \eqref{eq:tmp100} is from \eqref{eq:tmp333} and the
definition of Fisher information matrix; \eqref{eq:tmp101} is from
\eqref{eq:tmp101-0} with $\Vm = \left[{\xv\over\|\xv\|}\
\bar{\Vm}\right]$;
\eqref{eq:tmp102} holds since $\Fm^\T \xv = 0$; \eqref{eq:tmp103}
is from the fact $\Fm^\T\Fm = \Fm^\T \Vm \Vm^T \Fm = \Fm^\T
(\xv\xv^\T/\|\xv\|^2) \Fm + \Fm^\T \bar{\Vm} \bar{\Vm}^\T \Fm =
\Fm^\T \bar{\Vm} \bar{\Vm}^\T \Fm$. 
Therefore, $\det(\Jm_{y,\hv|{\xv}}(\tilde{\xv})) = \EE \left[
\xi(\sqrt{\snr} S) \right]^{\nt-1} \det \left( \Id +
{1\over x_{\nt}^2} \tilde{\xv} \tilde{\xv}^\T \right) = \EE \left[
\xi(\sqrt{\snr} S) \right]^{\nt-1} (1 +
{\|\tilde{\xv}\|^2 \over x_{\nt}^2}) = {1 \over
1-\snr^{-1}\|\tilde{\xv}\|^2}\EE \left[
\xi(\sqrt{\snr} S) \right]^{\nt-1}$. Taking the integral over the
space $\sqrt{\snr}\mathcal{B}_{\nt-1}$, we have 
\begin{IEEEeqnarray}{rCl}
  \int_{\sqrt{\snr} \mathcal{B}_{\nt-1}} {1\over\sqrt{1-
  \snr^{-1}\| \tilde{\xv} \|^2}}
  \d\tilde{\xv}  &=& \snr^{\nt-1\over2}\Vol(\mathcal{B}_{\nt-1}) \EE_{\mathcal{B}_{\nt-1}}\left[ {1\over\sqrt{1-r^2}} \right] \\
  &=& \snr^{\nt-1\over2} \Vol(\mathcal{B}_{\nt-1}) \int_{0}^1 
  {(\nt-1)r^{\nt-2}\over\sqrt{1-r^2}} \d r \\
  &=& \snr^{\nt-1\over2} {\Gamma({\nt+1\over2})\over\Gamma({\nt\over2})} \sqrt{\pi} \Vol(\mathcal{B}_{\nt-1}) \\
  &=& \snr^{\nt-1\over2} {\pi^{\nt\over2}\over\Gamma({\nt\over2})} \\
  &=& {\snr^{\nt-1\over2} \over2} \Vol(\mathcal{S}_{\nt-1}),
\end{IEEEeqnarray}%
where we applied results from Lemma~\ref{lemma:vol_ball}. 
\qed

\bibliographystyle{IEEEtran}
\bibliography{IEEEabrv,./biblio}

\end{document}